\documentclass[a4paper,english,cleveref, autoref, nolineno]{lipics-v2019}

\title{Probing a Set of Trajectories to Maximize Captured Information}

\author{S\'andor P.~Fekete}{Department of Computer Science, TU Braunschweig, Germany}{s.fekete@tu-bs.de}{}{}
\author{Alexander Hill}{Department of Computer Science, TU Braunschweig, Germany}{a.hill@tu-bs.de}{}{}
\author{Dominik Krupke}{Department of Computer Science, TU Braunschweig, Germany}{d.krupke@tu-bs.de}{}{}
\author{Tyler Mayer}{Decision Management Systems, Charles River Analytics Inc., Boston, USA.}{tmayer@cra.com}{}{}
\author{Joseph S.~B.~Mitchell}{Department of Applied Mathematics and
Statistics, Stony Brook University, USA.}{joseph.mitchell@stonybrook.edu}{}{}
\author{Ojas Parekh}{Sandia National Laboratories, USA.}{odparek@sandia.gov}{}{}
\author{Cynthia A.~Phillips}{Sandia National Laboratories, USA}{caphill@sandia.gov}{}{} 
\authorrunning{S.P. Fekete, A. Hill, D. Krupke, T. Mayer, J.S.B. Mitchell, O. Parekh, C.A. Phillips}

\Copyright{S.P. Fekete, A. Hill, D. Krupke, T. Mayer, J.S.B. Mitchell, O. Parekh, and C.A. Phillips}



\category{}

\relatedversion{}

\supplement{}

\nolinenumbers 

\hideLIPIcs  

\usepackage{graphicx}
\usepackage[binary-units=true]{siunitx}
\usepackage[export]{adjustbox}
\usepackage{amsmath}
\usepackage{mathtools} 
\usepackage{subcaption}
\usepackage{algpseudocode}
\usepackage{xcolor}
\usepackage{todonotes}

\newcommand{\old}[1]{{}}   
\newcommand{\remove}[1]{{}}   
\newcommand{\cindy}[1]{\textcolor{green!60!black}{(Cindy: #1)}}
\newcommand{\joe}[1]{\textcolor{blue}{(Joe: #1)}}

\newcommand{\sandor}[1]{\textcolor{red}{(Sandor: #1)}}
\newcommand{\ojas}[1]{\textcolor{violet}{(Ojas: #1)}}
\newcommand{\move}[1]{\color{magenta}{#1}\color{black}}

\renewcommand{\move}[1]{}

\newif\ifabstract
\abstracttrue
\abstractfalse  %

\newif\iffull
\ifabstract \fullfalse \else \fulltrue \fi

\begin{document}

\maketitle

\begin{abstract} \large\baselineskip=12pt 
  We study a trajectory analysis problem we call the \textsc{Trajectory
    Capture Problem} (TCP), in which, for a given input set ${\cal T}$
  of trajectories in the plane, and an integer $k\geq 2$, we seek to
  compute a set of $k$ points (``portals'') to maximize the total
  weight of all subtrajectories of ${\cal T}$ between pairs of
  portals.  This problem naturally arises in trajectory analysis and
  summarization.

  \baselineskip=12pt
  We show that the TCP is NP-hard (even in very special cases) and give
  some first approximation results.  Our main focus is on attacking
  the TCP with practical algorithm-engineering approaches, including
  integer linear programming (to solve instances to provable
  optimality) and local search methods. We study the integrality gap
  arising from such approaches. We analyze our methods on
  different classes of data, including benchmark instances that we
  generate.  Our goal is to understand the best performing heuristics,
  based on both solution time and solution quality.  We demonstrate
  that we are able to compute provably optimal solutions for
  real-world instances.
\end{abstract}

\section{Introduction}
In recent years, the progress in technical capabilities has resulted in
massive amounts of trajectory data for cars, trucks, trains, aircraft,
ships, people, and animals being collected at increasing rates. 
This presents major challenges for storing and evaluating this ever-growing
data, as well as for extracting useful information; this motivates
the search for data structures and algorithms that capture some of the most important
and useful aspects of such trajectories. At the same time, the 
availability of large volumes of data makes it possible to consider
useful aspects that were previously unavailable due to the lack of data
or algorithmic evaluation methods, such as collecting useful information
along the traveled trajectories.

One such means of analyzing a set ${\cal T}$ of trajectories is to determine 
``popular pairs'' $(p_1,p_2)$ of locations (or location/time pairs) for which
there is a significant ``value'' of the trajectories ${\cal T}$ going between
those points. This value can arise from aggregated data between checkpoints, such
as the total passenger-distance or the accumulated total pollution along the way;
it also comes up through the use of tomographic methods 
(i.e., determining physical phenomena by measuring aggregated effects along the path between two sensors), 
which are highly important in the context of many other application areas, such as in
astrophysics~\cite{korth2002particle}. A limiting factor is usually the need to pick a set of locations of finite cardinality, 
i.e, placing a limited number of toll booths, cameras for average speed measurement, various other
types of sensors, or abstract collections of focus points for sampling trajectories.
For an example, consider the scenario shown in Figure~\ref{fig:sanfran},
which corresponds to more than 500,000 data points that arise from the trajectories
of over 250 taxi cabs in San Francisco. Our goal is to
identify a small subset of locations that allow us to capture as much
of the movement information, in terms of weighted distance between checkpoints, as possible.

\begin{figure}[h]
\centering
\includegraphics[width = 0.7\textwidth]{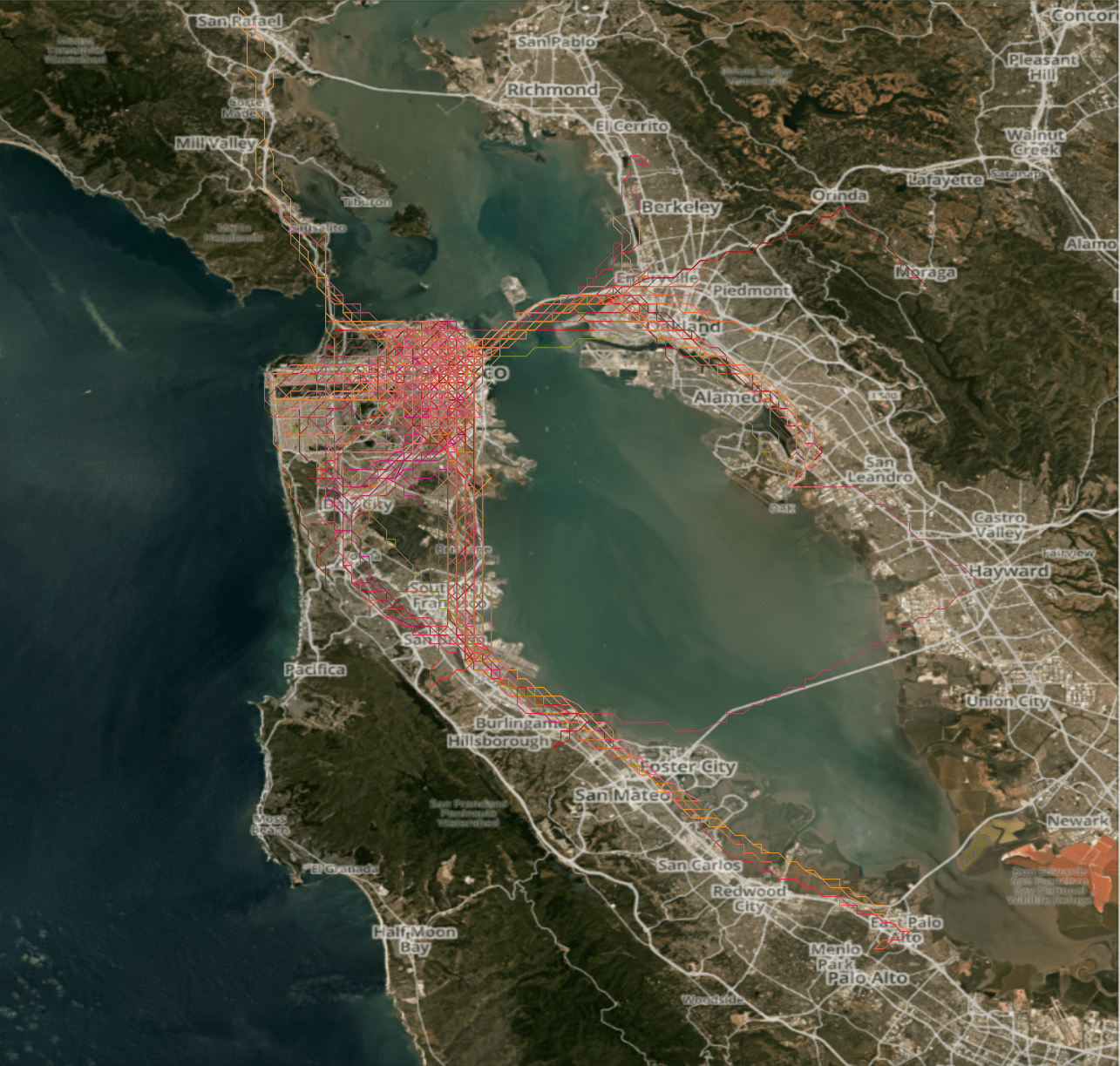}
  \caption{A set of taxi trajectories after preprocessing in San Francisco Bay Area.\\ (Satellite images courtesy of Planet Labs Inc.)}
\label{fig:sanfran}
\end{figure}

In this paper, we study the {\sc Trajectory Capturing Problem} (TCP), a generalization of ``popular pair''
computation: Given a set ${\cal T}$ of trajectories and an integer $k\geq
2$, determine a set of $k$ {\em portals} (points) to maximize the sum
of the weights of the inter-portal subtrajectories in ${\cal T}$; such
subtrajectories are said to be ``captured'' by the set of portals.
After establishing that the TCP is NP-hard, and giving some first
approximation bounds, we focus on algorithm
engineering methods for attacking the problem practically.

\subsection{Our Results}

We provide results both for algorithmic theory and algorithm
engineering aspects of the \textsc{Trajectory Capturing Problem}.

\begin{itemize}
\item We prove NP-hardness for the TCP, even 
for instances with trajectories consisting of individual axis-parallel segments.
\item We establish two approximation algorithms. One has approximation factor $K$ if the input trajectory set decomposes into $K$ subsets where within each subset no two segments cross, though they can overlap, ($K$ noncrossing subsets). The other has an approximation factor $\Delta$, the maximum number of input trajectories hit by any single point.
\item We develop an Integer Linear Program (IP) to solve TCP instances
  to provable optimality.
\item We show that, in general, the IP formulation has unbounded
  integrality gap\footnote{The integrality gap is $\max_{\cal I} \frac{\mbox{LP}({\cal I})}{\mbox{IP}({\cal I})}$, where ${\cal I}$ is a TCP instance, IP$({\cal I})$ is the optimal IP solution value and LP$({\cal I})$ is the optimal solution value to the linear programming (LP) relaxation.}. 
  For inputs
  decomposing into 2 noncrossing subsets (e.g., arising from
  axis-parallel segments), we show that the integrality gap is at most
  $\frac{k}{\lfloor k/2\rfloor}$
  (Theorem~\ref{thm:integrality-gap}).
\item We develop methods for generating challenging 
benchmark instances for experimentation. For geometric instances, based on segment arrangements, this requires care to address geometric robustness and accuracy.
\item We compute provably optimal solutions for general instances 
up to thousands of candidate capture points.
\item We give provably optimal solutions for even larger instances, with up to 7000
possible capture points, for instances based on axis-parallel segments, where we find the integrality gap to be quite small.
\item Using the IP solutions as a reference, we perform 
a thorough computational study using heuristic algorithms (a greedy algorithm, iterated local search, simulated annealing, and an evolutionary algorithm), with various settings, to understand how heuristics perform on various instances, 
in terms of time and solution quality.
\item We demonstrate our methods on 
real-world instances, including a provably optimal solution 
for taxi-cab data on 250 trajectories, with more than 
500,000 individual geographic data points.
\end{itemize}

Given the broad range of potential applications, scenarios and assumptions, 
we do not claim (or even aim) to provide a final set of methods for the 
general problem. Instead, we focus on demonstrating that a number of
modeling and optimization approaches can provide a promising range of insights and
tools for future work from various directions.

\subsection{Related Work}
Related to our problem is the well-studied \textsc{Geometric Hitting Set Problem} (GHS), in which
one seeks a smallest cardinality set of points to hit a given set of
lines, segments, or trajectories in the plane; as a single point
suffices to capture all of a trajectory it lies on, achieving large objective values for the 
GHS is easier than for the TCP.
The GHS is known to be NP-hard, hard to approximate (below a threshold), and
some natural geometric cases have constant-factor approximation algorithms;
see, e.g., \cite{fekete2018geometric}, and the references therein. 

There is a vast literature on problems of analyzing, clustering, mining, and
summarizing a set of trajectories.
For an extensive survey of trajectory data mining methods, see Zheng~\cite{zheng2015trajectory} and Zheng and Zhou~\cite{zheng2011computing}.
Notions of ``flocks'' and ``meetings'' have been formalized and studied
algorithmically~\cite{benkert2008reporting,gudmundsson2006computing,laube2008decentralized,vieira2009line}.
Gudmundsson, van Kreveld, and Speckmann \cite{gudmundsson2007efficient}
define \emph{leadership}, \emph{convergence}, and \emph{encounter} and provide exact
and approximate algorithms to compute each.  Andersson, Gudmundsson,
Laube, and Wolle \cite{andersson2008reporting} show that several {\em
  Leader-Problem} (LP) variants ({\em LP-Report-All, LP-max-Length,
  LP-Max-Size}) are all polynomially solvable and provide exact
algorithms.  Buchin et al. \cite{buchin2013trajectory} present a
framework to fully categorize trajectory grouping structures
(grouping, merging, splitting, and ending of groups).  
Other approaches to trajectory summarization naturally include
cluster analysis, of which there is a large body of related work.
Li, Han, and Yang \cite{li2004clustering}
consider rectilinear trajectories and show how to cluster with
bounding rectangles of a given size.  Several approaches (e.g.,
\cite{hadjieleftheriou2003line, lee2008traclass, lee2007trajectory, 
nanni2006time}) consider density-based methods for
clustering sub-trajectories. Lee, Han, Li, and Gonzalez \cite{lee2008traclass} take it one
step further by considering a two-level clustering hierarchy that first
accounts for regional density and then considers lower-level 
movement patterns.  Li, Ding, Han, and Kays \cite{li2010swarm} consider a problem
(related to \cite{gudmundsson2007efficient}) in which they seek to identify all
{\em swarms} or groups of entities moving within an arbitrary shaped
cluster for a certain, possibly disconnected, duration of time.
Also, Uddin, Ravishankar, and Tsotras \cite{uddin2011finding} consider finding what
they call {\em regions of interest} in a trajectory database.

In motivating tomographic applications, the number of checkpoints is an important constraint in the use of discrete tomography, e.g., in astrophysics (Korth et al.~\cite{korth2002particle}).

\section{Preliminaries}
\label{sec:preliminaries}

We are given a set of trajectories ${\cal T}$, each specified by a
sequence of
points, e.g., in the Euclidean plane.
We seek a set $P=\{p_1,\ldots,p_k\}$ of $k$ {\em portals}, i.e., selected points 
that lie on some of the trajectories.
While our practical study focuses on instances in which the trajectories
${\cal T}$ are purely spatial, e.g., given as polygonal chains or line segments in the plane,
our methods apply equally well to more general portals and to trajectories that include a temporal component and live in space-time.
More generally, we are given a graph ${\cal G}$, with length-weighted edges, and a set of paths within ${\cal
 G}$. We wish to determine a subset of $k$ of the nodes of ${\cal G}$ that
maximizes the sum of the (weighted) lengths of the
subpaths (of the input paths) that link consecutive portals along the
input paths. 

We seek to compute a $P$ that maximizes the total {\em captured
  weight} of subtrajectories between pairs of portals.  
For a trajectory $\tau\in {\cal T}$, if there are two or more portals
of $P$ that lie along $\tau$, say $\{p_{i_1},\ldots,p_{i_q}\}$ (for
$q\geq 2$), then the subtrajectory, $\tau_{p_{i_1},p_{i_q}}$, between
$p_{i_1}$ and $p_{i_q}$ is {\em captured} by $P$, and we get credit
for its weight $f(\tau_{p_{i_1},p_{i_q}})$.
(For many of our instances, $f(\tau_{p_{i_1},p_{i_q}})$ corresponds
to the Euclidean distances, denoted by $|\tau_{p_{i_1},p_{i_q}}|$,
but our methods generalize to other types of weights.)
Let $f_P(\tau)$ denote the captured weight of trajectory $\tau$ by the
portal set $P$. The \textsc{Trajectory Capture Problem} (TCP) is then to
compute, for given ${\cal T}$ and $k$, a set of $k$ portals
$P=\{p_1,\ldots,p_k\}$ to maximize $\sum_{\tau\in {\cal T}}
f_P(\tau)$.

\section{Analytical Results}\label{sec:theory}

The TCP is NP-hard and hard to approximate for general graphs even when all trajectories have weight $1$.
Given a graph, let each edge be a trajectory.
Then an optimal solution to TCP gives the densest $k$-node subgraph.
Manurangsi~\cite{manurangsi2017almost} showed that, assuming the exponential time hypothesis, there cannot be an $n^{\frac{1}{{\log \log n}^c}}$-approximation algorithm for the \textsc{Densest k-Subgraph} problem for any constant $c>0$.

In the following, we give more specific results for a range of geometric TCP versions.

\subsection{One-Dimensional TCP}

In the one-dimensional setting, the underlying graph ${\cal G}$ is a path, and the input trajectories 
${\cal T}=\{ (a_1,b_1),\ldots,(a_n,b_n)\}$ are a set of
subpaths of ${\cal G}$, specified by pairs of integers, $a_i,b_i$.  
A solution to the TCP then consists of $k$ points, $P=\{p_1,\ldots,p_k\}$,
w.l.o.g.\ indexed in sorted order, $p_1<p_2<\cdots <p_k$.  

\begin{theorem}
  The one-dimensional TCP can be solved exactly in polynomial time.
  \end{theorem}
\begin{proof}
  For $i=1,2,\ldots,k-1$, let $V_i(x)$ be the maximum possible length of ${\cal T}$ captured by points $(p_i,\ldots,p_k)$, with $p_i=x\in \{a_1,\ldots,a_n,b_1,\ldots,b_n\}$; let $V_k(x)=0$, for any $x$.
  Then, the value functions $V_i$ satisfy the following dynamic programming recursion, for $i=1,2,\ldots,k-1$, and each
  $x\in \{a_1,\ldots,a_n,b_1,\ldots,b_n\}$:
  $$V_i(x) = \max_{x'\in \{a_1,\ldots,a_n,b_1,\ldots,b_n\}, x'>x} \{ V_{i+1}(x')+ \sum_{j: (x,x')\subseteq (a_j,b_j)} (x'-x) \}.$$
  The summation counts the length $(x'-x)$ once for each input interval that contains the interval $(x,x')$.
  We can compute the $O(nk)$ values $V_i(x)$ in time $O(n^2k)$ by incrementally updating the summation as we consider values of $x'$ in increasing order.
\end{proof}

\subsection{Two-Dimensional TCP}

\subsubsection{Complexity}

\begin{theorem}
  \label{thm:general-hard}
The TCP is NP-hard, even for an input ${\cal T}$ of $n$ line segments in the plane.
\end{theorem}

The proof of Theorem~\ref{thm:general-hard} (in the Appendix) uses a
construction involving segments of \emph{many} orientations whose pairwise
intersections may only be single points. The following shows that
the TCP is already NP-hard for segments of \emph{two} orientations,
provided that two intersecting segments may be collinear, and three
different segments can intersect in a single point.

\begin{theorem}
  \label{thm:2orientation-hard}
The TCP is NP-hard, even for an input ${\cal T}$ of $n$ 
line segments of at most two orientations in the plane, with any two segments intersecting in at most a single point (there are no overlapping pairs of segments) but possible collinearity.
\end{theorem}

\subsubsection{Approximation Algorithms}

Consider first the case in which the set ${\cal T}$ of input
trajectories is the (disjoint) union of $K$ subsets of trajectories,
${\cal T}={\cal T}_1\cup {\cal T}_2\cup \cdots \cup {\cal T}_K$, with
each subset ${\cal T}_i$ having the following {\em path property}:
For any connected component of the intersection graph of ${\cal T}_i$,
the trajectories in that component are all subpaths of some path in the union of ${\cal T}_i$.
This condition holds, for example, if ${\cal T}$ is a set of line segments of
$K$ distinct orientations.

\begin{theorem}
The TCP for an input set ${\cal T}={\cal T}_1\cup {\cal T}_2\cup
\cdots \cup {\cal T}_K$, with each subset ${\cal T}_i$ having the path property,
has a polynomial-time $K$-approximation algorithm.
\end{theorem}

\begin{proof}
Within each class ${\cal T}_i$, the path property implies that the trajectories behave like intervals on a
line, so our one-dimensional dynamic programming solution applies.  Selecting the best solution that uses all $k$ points for one of the ${\cal T}_i$ gives
a polynomial-time algorithm with approximation ratio $K$.
\end{proof}

\begin{theorem}
  The TCP for an input set ${\cal T}$ of arbitrarily
  overlapping/crossing trajectory paths in the plane having bounded
  depth $\Delta$ (i.e., no point of $\Re^2$ lies in more than $\Delta$ input
  trajectories) has a polynomial-time $\Delta$-approximation algorithm, for any even number, $k$, of portals.
  (If $k$ is odd, the approximation factor is at most $\Delta(1+\frac{1}{k-1})$.)
\end{theorem}

\begin{proof}
  Consider an optimal set, $H^*$, of $k$ hit points, capturing total
  trajectory length $L^*$. For each hit point $h\in H^*$, which lies
  on $\delta\leq \Delta$ trajectories of ${\cal T}$, we replace $h$
  with $\delta$ copies (``clones'') of the point $h$, with one clone
  associated with each of the $\delta$ trajectories that $h$ hits.  In
  total there are at most $k\Delta$ clones.  Consider any trajectory
  $\tau\in{\cal T}$, and consider the clones (copies of hit points)
  that lie along $\tau$.  If there are at least 2 clones on $\tau$,
  then the portion of $\tau$ that lies between the extreme clones on
  $\tau$ is captured; the length of this portion is at most the length
  of $\tau$. The $k\Delta$ clones capture portions of at
  most $\lfloor k\Delta/2 \rfloor$ trajectories, resulting in total
  captured length $L^*\leq \ell_1+\ell_2+\cdots + \ell_{\lfloor
    k\Delta/2 \rfloor}$, where $\ell_i$ denotes the length of the
  $i$th longest trajectory of ${\cal T}$ ($\ell_1\geq \ell_2\geq
  \ell_3\geq \cdots$).

  Now consider the simple greedy algorithm that places
  hit points at the 2 endpoints of the $\lfloor k/2\rfloor$ longest
  trajectories of ${\cal T}$, using at most $k$ total hit points.  This algorithm 
  captures length $L=\ell_1+\ell_2+\cdots + \ell_{\lfloor k/2\rfloor}$.
  The approximation ratio is at most
    \[\frac{L^*}{L} \leq { \lfloor k\Delta/2\rfloor  \over  \lfloor k/2\rfloor }.\]
  Thus, $\frac{L^*}{L} \leq \Delta$ for even $k$. For odd $k$, the denominator
  is exactly $(k-1)/2$, while the numerator is either $k\Delta/2$ (if $\Delta$ is even) or $(k\Delta -1)/2$ (if $\Delta$ is odd); thus, $\frac{L^*}{L} \leq \frac{k\Delta/2}{(k-1)/2} = \Delta(1 + \frac{1}{k-1})$.
\end{proof}
%

\section{Algorithm Engineering}\label{sec:practice} 
As the TCP can be considered an optimization problem on a weighted graph, 
we can use approaches such as Integer Linear Programming and local search heuristics. Given the geometric origins of
the TCP, we consider geometric aspects; in addition, dealing with geometric
data involved a number of other aspects of algorithm engineering,
such as accuracy and correctness when handling locations, coordinates, and
intersections.

\subsection{Integer Linear Programming}
\label{sec:IP}

\subsubsection{An IP Formulation}
As a problem of combinatorial optimization, the TCP can be modeled
as an Integer Linear Program (IP), for which solutions can be computed
with the help of powerful IP solvers.  The following IP models the TCP.

\medskip
$
\begin{array}{llr}
\medskip
  &\max  \sum_{\tau \in \mathcal{T}, e\in E(\tau)}f(e) x_{\tau, e}\\
\smallskip
  &\hspace*{.6cm}\sum_{v\in V} y_v \leq  k & \text{(Constraint 1)}  \\
  \forall \tau=(v_0, \ldots, v_l)\in \mathcal{T}:\nonumber\\
  \ \ \ \ \ \forall i\in 0,\ldots, l-1:&  \left\{ \begin{array}{ccll}
    	x_{\tau, v_iv_{i+1}} & \leq & y_i \quad & \text{if $i=0$,} \\
    	x_{\tau, v_iv_{i+1}} & \leq & y_i + x_{\tau, v_{i-1}v_i} \quad   & \mbox{ else} \\
  \end{array}\right.&\text{(Constraint 2)}  \\
  \ \ \ \ \ \forall i\in 1,\ldots, l:&  \left\{ \begin{array}{ccll}
    	x_{\tau, v_{i-1}v_{i}} &\leq& y_i \quad &\text{if $i=l$,} \\
    	x_{\tau, v_{i-1}v_{i}} &\leq& y_i + x_{\tau, v_{i}v_{i+1}} \quad   &\mbox{ else} \\
  \end{array}\right. & \text{(Constraint 3)} \\
\smallskip
  \forall v\in V, \tau\in e\in \tau: &\omit\rlap{\hspace*{.6cm} $x_{\tau, e}, y_v \quad \in \quad \{0,1\} $}
\end{array}
$

\medskip
We have two types of Boolean variables: $y_v$, for $v\in V$, which indicates if
node $v$ is one of the $k$ selected portals, and $x_{\tau, e}$, for edge $e\in
E$ on trajectory $\tau$, which indicates if the portion $e$ of trajectory
$\tau$ is captured by selected portals.  For an edge $e$, there are distinct
variables, $x_{\tau, e}, x_{\tau',e}$, for trajectories $\tau\neq \tau'$,
because $e$ can be captured in $\tau$ but not in $\tau'$.

Our objective function  maximizes the weighted sum of captured trajectory edges, where $E(\tau)$ denotes the edges of $\tau$ in $\mathcal{G}$, and $f(e)$ is the weight (i.e. length) of edge $e$.
(Optionally, we could have trajectory-dependent weights on edges.)
Constraint~1 limits the number ($\leq k$) of selected portals.
Constraints~2 and 3 enforce that, in order for an edge to be captured as part of trajectory $\tau$, there must be a selected portal in each direction; either there is a selected portal at the next node, or the following trajectory edge is also captured.
In the latter case, because $\tau$ has no cycle (it is a simple path), there must be a selected portal on $\tau$ at some point in that direction if any portion of $\tau$ is to be captured.

For an example, see Section~\ref{sec:example} in the Appendix.

\subsubsection{Fractional Solutions}

Relaxing the integrality constraints of the IP may result in fractional
solutions.  We show (in the Appendix) that the gap (ratio) between the best
fractional and the integral (optimal) solution objective functions can be
arbitrarily large, for any fixed~$k$.

\begin{theorem}
  \label{thm:gap}
  The integrality gap for the TCP IP
  can be arbitrarily large for any~$k$.
\end{theorem}

For instances arising from non-overlapping (i.e., no parallel
segments may share more than one point) axis-parallel segments, we can bound
the integrality gap, because the particularly bad ``clusters'' of the general
case cannot occur.

\begin{theorem}
  \label{thm:integrality-gap}
  For trajectories ${\cal T}$ arising from  non-overlapping axis-parallel line
segments, the integrality gap is at most $\frac{k}{\lfloor k/2\rfloor}$, for
$k\geq 2$. 
\end{theorem}

\begin{proof}
  We can easily get an integral solution by simply capturing the $\lfloor
\frac{k}{2} \rfloor$ longest trajectories (segments) by selecting their at most
$k$ endpoints as portals.

  We create a new LP instance, called LP$_2$, by including two copies $v_1,v_2$
of each portal variable $v$, so that one copy lies only on horizontal segments,
while the other lies only on vertical segments.  We constrain
$y_{v_1}=y_{v_2}=y_v$ and allow a budget of $2k$ portals for LP$_2$.  Any
feasible values for the $y_i$ in the original LP solution are still feasible in
LP$_2$ (setting both copies), so the optimal solution of LP$_2$ is an upper
bound for the original LP.  Because segments do not overlap, every portal now
lies only on a single segment.
  Thus the optimal solution for LP$_2$ covers the $\frac{2k}{2}=k$ longest segments. This shows the integrality gap is at most $\frac{k}{\lfloor k/2\rfloor}$.
\end{proof}

The bound of Theorem~\ref{thm:integrality-gap} is tight for $k=2$: Consider four segments that are edges of a unit square; then, $k=2$ portals can capture at most length 1, while a fractional value of $1/2$ at each of the four corners yields objective value $4/2=2$ for the LP.  
For $k\geq 4$, it becomes increasingly difficult to build instances with a high integrality gap.

\subsection{Heuristics}
\label{Sec:Heu}

Integer Linear Programming solvers can provide provably optimal or near-optimal
solutions for relatively large instances. However, eventually runtime
and memory requirements become a limiting factor for large enough instances,
so it becomes important to develop effective heuristics. We considered
a spectrum of heuristics: 
\emph{Greedy}, which constructs solutions from scratch by locally optimal choices; 
\emph{Iterated Local Search}, which 
iteratively improves a current solution by finding a better one in its local neighborhood;
\emph{Simulated Annealing}, which uses a ``temperature'' function that governs
the probability of temporarily accepting a worse solution during a local search;
and \emph{Genetic Algorithms}, which maintain a selection of solutions that are
locally modified and combined to achieve gradually better solutions.

{\em Greedy} begins by selecting the two  portals at the ends of the
longest trajectory, and then incrementally, greedily selects portals
that in each step increase the total captured length as much as possible.
{\em Greedy} can be fooled and give poor solutions; it can,
though, serve to give a reasonable starting solution for our other
metaheuristics.

{\em Iterated Local Search} (ILS) is a basic metaheuristic that, given
an initial solution, iteratively replaces the current solution with
the best solution found by applying a single local modification, until
no further improvement can be achieved. The set of solutions that can
be obtained by a single local modification from a specific solution is
called its {\em neighborhood}. For a local modification operator based
on changing a single portal, the neighborhood consists of all
solutions that differ in exactly one portal. The smaller the
neighborhood, the faster the best solution within it can be found;
however, a smaller neighborhood also reduces the search space and
correspondingly can reduce the quality of the obtained solutions. We
considered {\em global
neighborhoods}, based on moving a random portal to an alternative
random candidate node, and {\em local neighborhoods}, based on moving
a single portal to positions adjacent to other (unmoved)
portals.
ILS is initialized with any reasonable solution; after some
experimentation with alternatives (e.g., random selection), we settled
on using {\em Greedy} as the starting solution for ILS.

{\em Simulated Annealing} is similar to ILS, but instead of searching
for the best solution in the neighborhood, it selects a random
solution in the neighborhood and moves to it if (1) it is
an improvement, or (2) if it is not an improvement, but it passes a random test. 
The probability of moving to a worse solution is
determined by a ``temperature function'' and decreases with
time. Initially, it can easily escape local optima;
when the search satisfies a
termination criterion, it returns the best solution it found.

We considered three different termination criteria: 
the total number of iterations, the number of iterations without an improvement,
and the total runtime. For temperature regulation, we used a geometric reduction
by a constant multiplicative amount; for diversification we ``reheated" the
temperature to the start temperature when we did not change the solution 
for a certain number of iterations. For translating the temperature function
to a probability function, we used the Bolzmann 
function: $\mathrm{bolzmann}(s',s'',T) := \exp ^{(- \frac{s'-s''}{T})}$, where $s'$ and $s''$ are the captured weight of two neighboring solutions.
In addition, we used parallelization for pursuing multiple searches from
different starting points.

{\em Evolutionary algorithms (EAs)} are motivated by the way
adaption to environmental conditions happens in nature.
They maintain a ``population'' of current solutions. At each step,
the EA produces new solutions through mutations 
(i.e., local changes) and recombination
(combining pieces of solutions in the current population to create new ones).
Then, the EA keeps the best solutions (previous or new) to maintain a stable population size. 
We create the initial population by a version of Greedy that starts with a random segment
instead of the longest one.
For mutations, we used ILS or SA.
The probability of selection for recombination is $\frac{f(s)-f(s_{\min})}{\sum_{s'\in S} f(s')}$.
We used uniform random crossover.

\subsection{Generating Benchmark Instances}
\label{Sec:Gen}
Our IP and heuristic methods apply to general sets of trajectories ${\cal T}$, 
given by spatiotemporal or combinatorial data.
We focus on geometric instances, most of which are based on line-segment trajectories. 
Instances based on random segments tend to be very easy to solve
because most vertices have degree 2. So we generated instances based
on a set of seed points and selected segments linking them, resulting
in arrangement graphs with multi-trajectory intersections and more complicated covering
graphs. Alternatively, we tested adding
new intersection points to the set of seed points when
incrementally constructing the arrangement.
For all methods, we used exact intersection point computations from the {\sc Computational Geometry
Algorithms Library} (CGAL)
to overcome problems of floating point precision for large instances.
We generated seed points randomly, using a variety of spatial
distributions, including uniform distributions, point sets from the
TSP benchmark library~\cite{TSPLIB}, and point sets with density
distributions based on light maps, corresponding to population
densities (see \cite{dhh+-cnpmp-17}).

\subsection{Experimental Evaluation}
\label{sSec:Test} 
All experiments were performed on a single Intel(R) Core(TM) i7-4770 ($4\times\SI{3.4}{\GHz}$) with \SI{32}{\giga\byte} and CPLEX (V12.7.1 with default settings), with a time limit of \SI{900}{\second}.
The code and data is available at \url{https://github.com/ahillbs/trajectory_capturing}.

\subsubsection{Integer and Linear Programming}
\label{IPComp}
We first consider the sizes of instances that our IP can solve to
optimality within a 900-second time limit, and
we consider which factors contribute to the difficulty of the
instance.

\begin{figure}
  \centering
  \begin{subfigure}[t]{0.49\textwidth}
    \includegraphics[width=\textwidth]{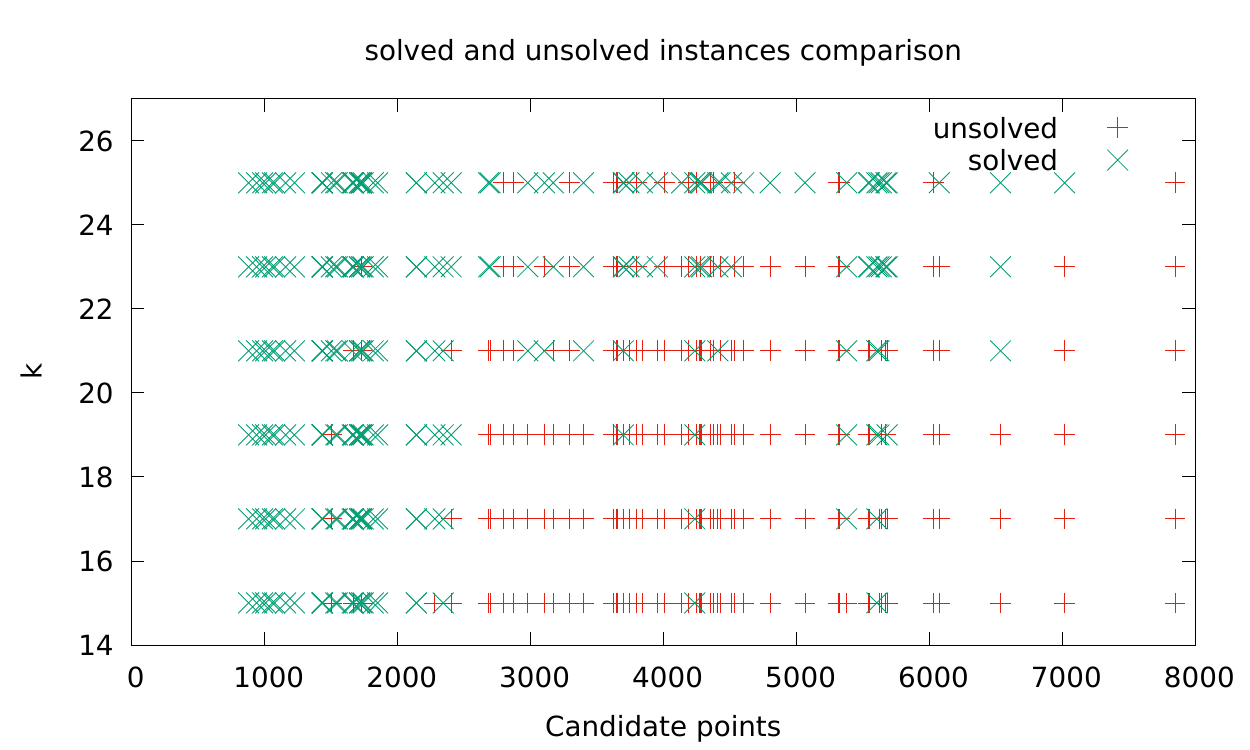}
    \subcaption{Solved and unsolved instances.}
    \label{fig:SolvedUnsolved_L}
  \end{subfigure}
  \begin{subfigure}[t]{0.49\textwidth}
    \includegraphics[width=\textwidth]{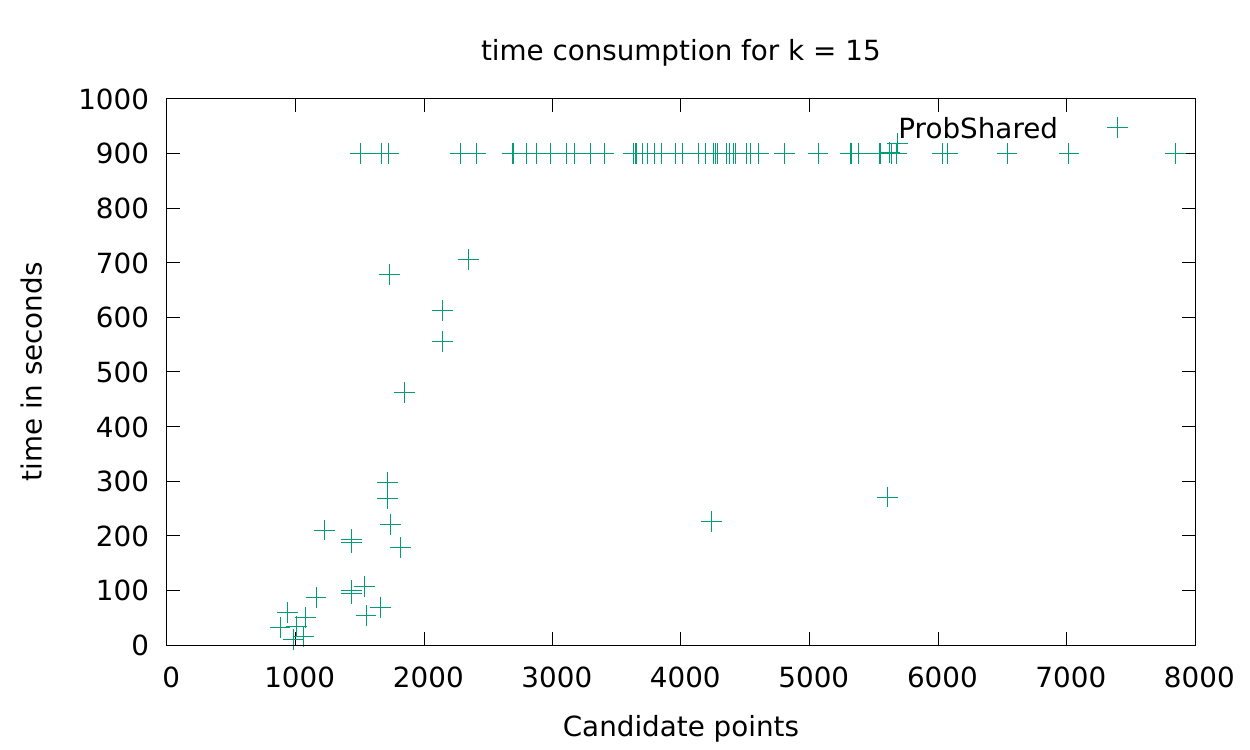}
    \subcaption{Time used for instances with $k = 15$}
    \label{fig:ProbK15Time_L}
  \end{subfigure}
  \caption[Probabilistic instances test results (large)]{
    Test results for solved and unsolved instances using the IP. The number of seed points varies from $35$ up to $55$ points. All tests were performed with a time limit of 900 seconds. }
    \label{fig:ProbResults_L}
\end{figure}

We varied the number of seed points between $35$ and $55$ and varied the (uniform) probability
of a segment connecting two seeds between $10\%$ and $20\%$.
In Figure~\ref{fig:SolvedUnsolved_L},
we see that instances with up to about $2500$ candidate points
(i.e., nodes at intersection points in the arrangement, where portals can be placed)
can be solved for $15
\leq k \leq 23$ to provable optimality within the time limit.
Instances with more than $2500$ candidate points are most often not solved for
$k < 23$.  For instances with $k \geq 23$, the problem seems to be easier to
solve.  
In Figure~\ref{fig:ProbK15Time_L} we see that for $ k = 15$ and
between $1500$ and $2000$ candidate points, instances start to become very
difficult to solve.  However, for $k \geq 23$, instances are
still solvable for more than $2500$ candidate points.

\subsubsection{Heuristic Methods}
\textbf{Neighborhoods for Local Methods.}
For modifying a given solution, we considered global neighborhoods, 
in which a portal is moved to an arbitrary other position, 
and local neighborhoods, in which a portal is only moved to positions that connect to another portal.
Using global neighborhoods, all solutions are theoretically quickly reachable
but they are significantly larger than local neighborhoods and, thus, a
meta-heuristic may not work in a focused enough way.
Details of this comparison can be found in Appendix~\ref{sec:local-global}; in particular
Figure~\ref{fig:NeighComp} shows our experimental evaluation.
Iterated Local Search yields the same solution quality for both neighborhoods;
with global neighborhood, only the runtime increases.
Simulated Annealing with global neighborhoods barely improves the initial greedy solution,
while it gives the best solutions with local neighborhoods.

As a result, we used local neighborhoods for all meta-heuristics.

\textbf{Mutation Strategy for Evolutionary Algorithms.}
For evolutionary algorithms, choosing the right kinds of mutations is of crucial importance,
as these allow reaching solutions that are not achievable only via recombination.
Practical usefulness requires focused mutations that have a high probability of being useful,
instead of purely random changes.
That is why we considered Iterated Local Search and Simulated Annealing as mutation operations.
As Simulated Annealing has a longer runtime, we used a faster terminating version (with potentially worse solutions) when using it for mutation.
In the following we refer to the version with Simulated Annealing as EASA and with Iterated Local Search as EAIS.
We have a start with 100 solutions and keep an ongoing population of 50 solutions.
The evolutionary algorithm stops after 15 minutes (but can take slightly longer to finish the last round) or if it has not found an improvement for multiple rounds.

\begin{figure}[h]
  \centering
  \begin{subfigure}[t]{0.49\textwidth}
    \includegraphics[width=\textwidth]{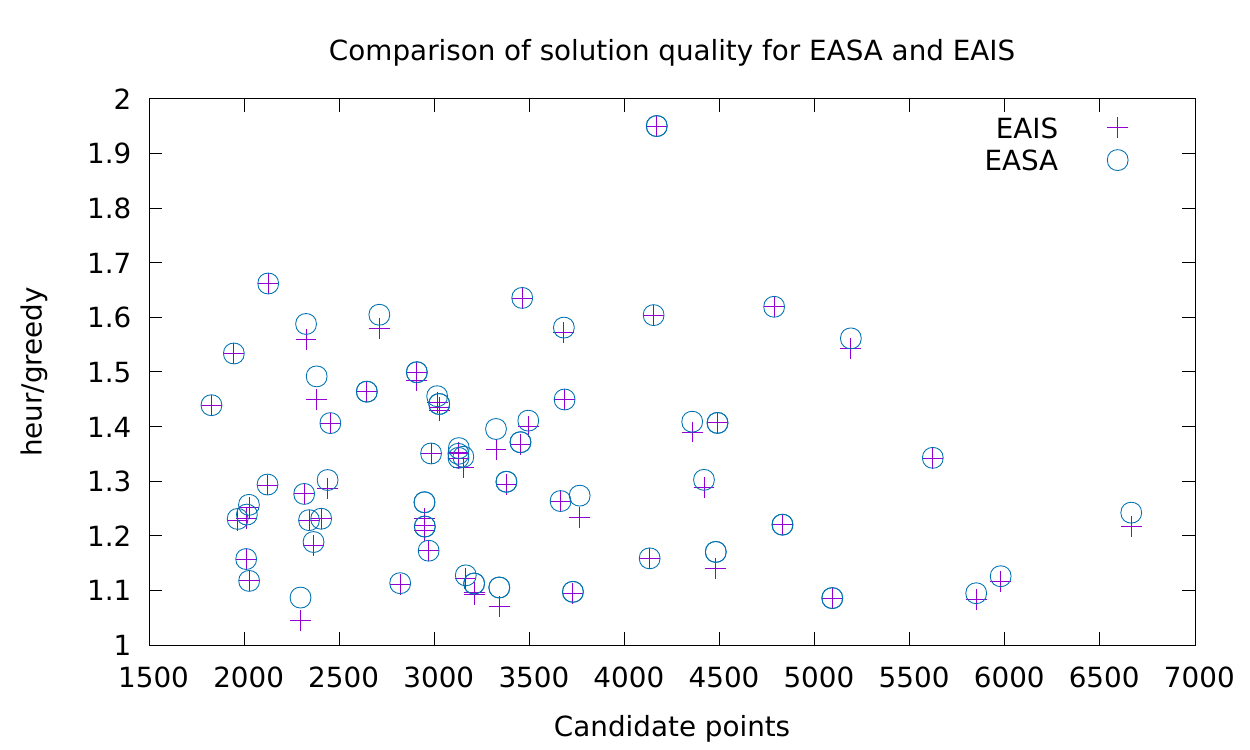}
    \subcaption{Solution quality for Evolutionary Algorithms}\label{fig:EvoCompSol}
  \end{subfigure}
  \begin{subfigure}[t]{0.49\textwidth}
    \includegraphics[width=\textwidth]{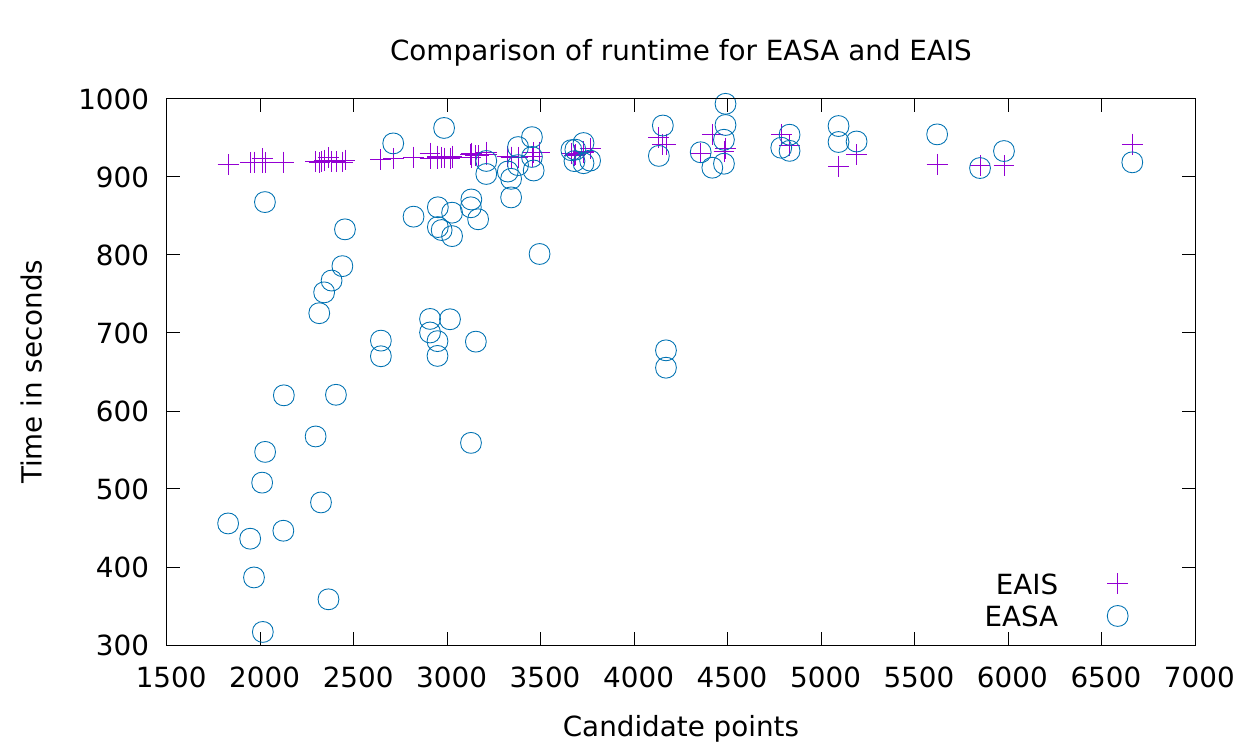}
    \subcaption{Runtimes for Evolutionary Algorithms}\label{fig:EvoCompTime}
  \end{subfigure}
  \caption[Comparison Evolutionary Algorithms]{Comparison of solution quality and runtime by Evolutionary Algorithm with Iterated Local Search and with Simulated Annealing.
  } 
  \label{fig:EvoComp}
\end{figure}

Figure~\ref{fig:EvoComp} shows the experimental comparison of both mutation variants.
One can see that EASA performs slightly better and is significantly faster for smaller instances.
This implies that EASA often quickly finds a good solution but is usually not able to improve it further and terminates early.
EAIS, on the other hand, is able to improve its quality until the time limit but still remains slightly worse.
For the further experiments, we settled on EASA.

\subsubsection{Comparison of Heuristics with IP as Baseline}
\label{IPvsHeur}
We compared the heuristics in terms of solution quality and runtime
against the IP solver, which produces not only solutions, but also guaranteed
bounds.

\begin{figure}
  \centering
  \begin{subfigure}[t]{0.49\textwidth}
    \includegraphics[width=\textwidth]{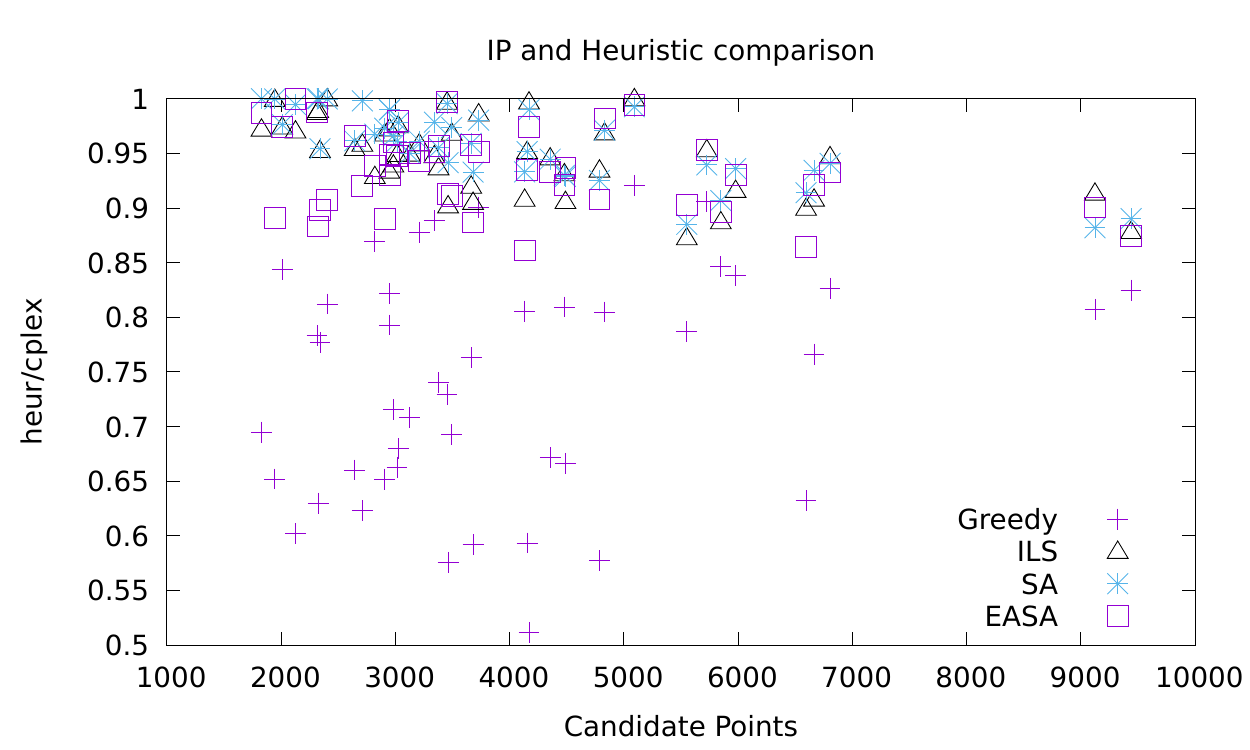}
    \subcaption{Solution quality for heuristics and IP solutions.}\label{fig:HeurCompSolAll}
  \end{subfigure}
  \begin{subfigure}[t]{0.49\textwidth}
    \includegraphics[width=\textwidth]{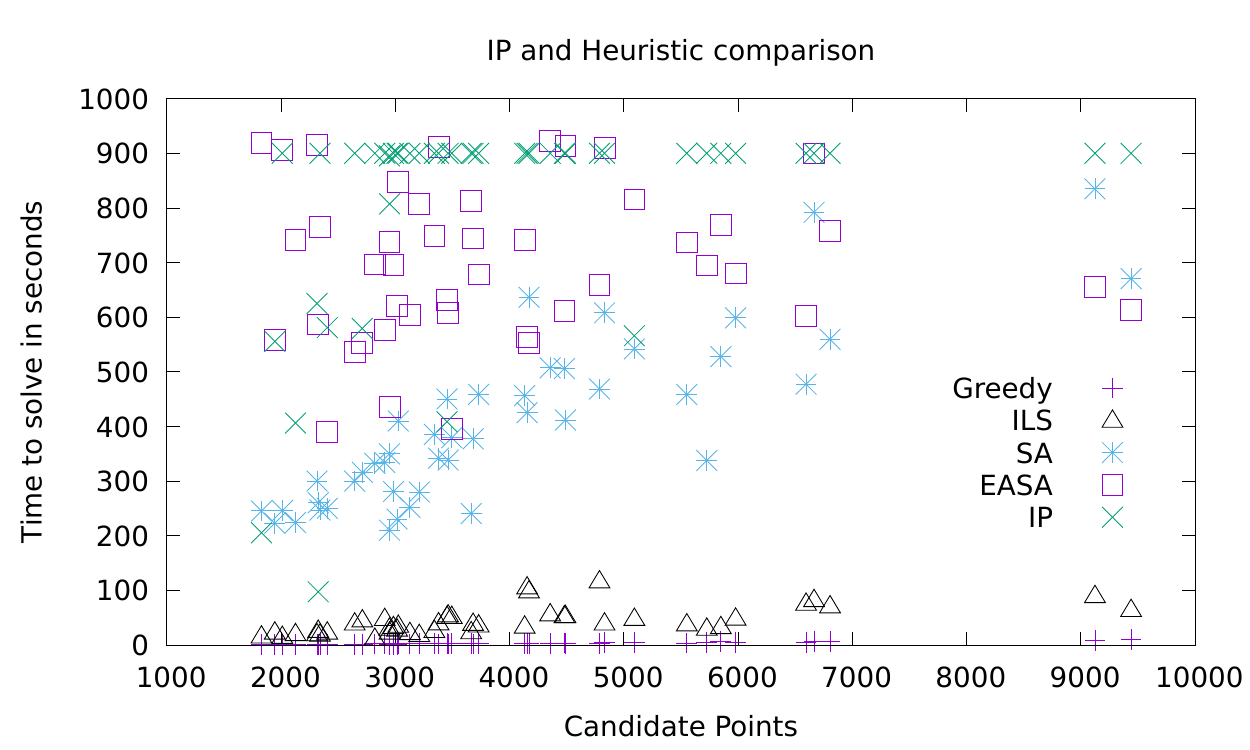}
    \subcaption{Runtimes of heuristics and IP solutions.}\label{fig:HeurCompTimeAll}
  \end{subfigure}
  \caption[Comparison all tested algorithms]{Comparisons of solution quality and runtime for all tested algorithms. 
}\label{fig:HeurCompAll}
\end{figure}

Figure~\ref{fig:HeurCompAll} shows the obtained results.
The Evolutionary Algorithm produces, on average, the worst solutions
of all metaheuristics, while still requiring more time than the others.  
However, it is the only metaheuristic tested that reliably computes good
solutions for instances consisting of several point clusters, for which solutions consist of several connected components.
These instances did not occur with our generation method but only in separate, manually created instances which are not part of this experiment.

The greedy approach never falls 
below $\frac{1}{2}$OPT for these instances, while being the fastest.

Iterated Local Search (ILS) appears to produce quite good solutions, which are not worse than 
10\% below the optimal solution, in a very short time frame.  While it only finds 
a local optimum, it seems that the objective function quality of these 
are quite close to the optimal values. As a consequence,
Iterated Local Search can produce good solutions for instances 
with up to $5000$ candidate points and $k = 25$.

The best heuristic algorithm, in terms of solution quality and runtime, appears to be 
Simulated Annealing.  The combination of fast diversification at high temperature
and random swaps for improving solutions at low temperature seems to work quite well;
in addition, the mechanism of 
reheating to restart diversification, in combination with
multi-threading to cover a larger search space, 
are characteristics that are not present in the Evolutionary Algorithm.
For $k = 25$, Simulated Annealing produces excellent solutions
for instances with up to $5000$ candidate points; this instance size
can be easily increased for smaller~$k$.

In summary, we can produce excellent heuristic solutions
for instances with up to $5000$ candidate points and $k = 25$.
If fast solutions are desired, Iterated Local Search is the method of choice.
The best tradeofff between runtime and solution quality is offered
by Simulated Annealing.
Finally, for cluster-based instances, we recommend Evolutionary Algorithms.

\subsubsection{Linear Programming and Integrality Gap}

As described in Section~\ref{sec:theory}, if the TCP input trajectories come from $K$ subsets of noncrossing trajectories, we have a $K$-approximation algorithm based on dynamic programming. In particular, if the input trajectories consist of axis-parallel line segments, $K=2$, so there is a 2-approximation.
This may coincide with better practical solvability
of these kinds of instances. We have verified this
for some instances for which all segments are axis-parallel
and (for collinear segments) non-overlapping. 
(See Figure~\ref{fig:huge} for such an instance with 1100 segments and
roughly 8200-8500 candidate points. We have also solved instances with 2000 segments and 19,000 candidate points.)
Furthermore, for instances with up to 20,000
points, the integrality gap was never larger than 20\% for $k=5$, 
7\% for $k=10$, 5\% for $k=15$ and less than 2.5\% for larger $k$.
See Figures~\ref{fig:TANOC5_large}--\ref{fig:TANOC25_large} in the Appendix.

\old{
\begin{figure}
  \centering
  \begin{subfigure}[t]{0.49\textwidth}
    \includegraphics[width=\textwidth]{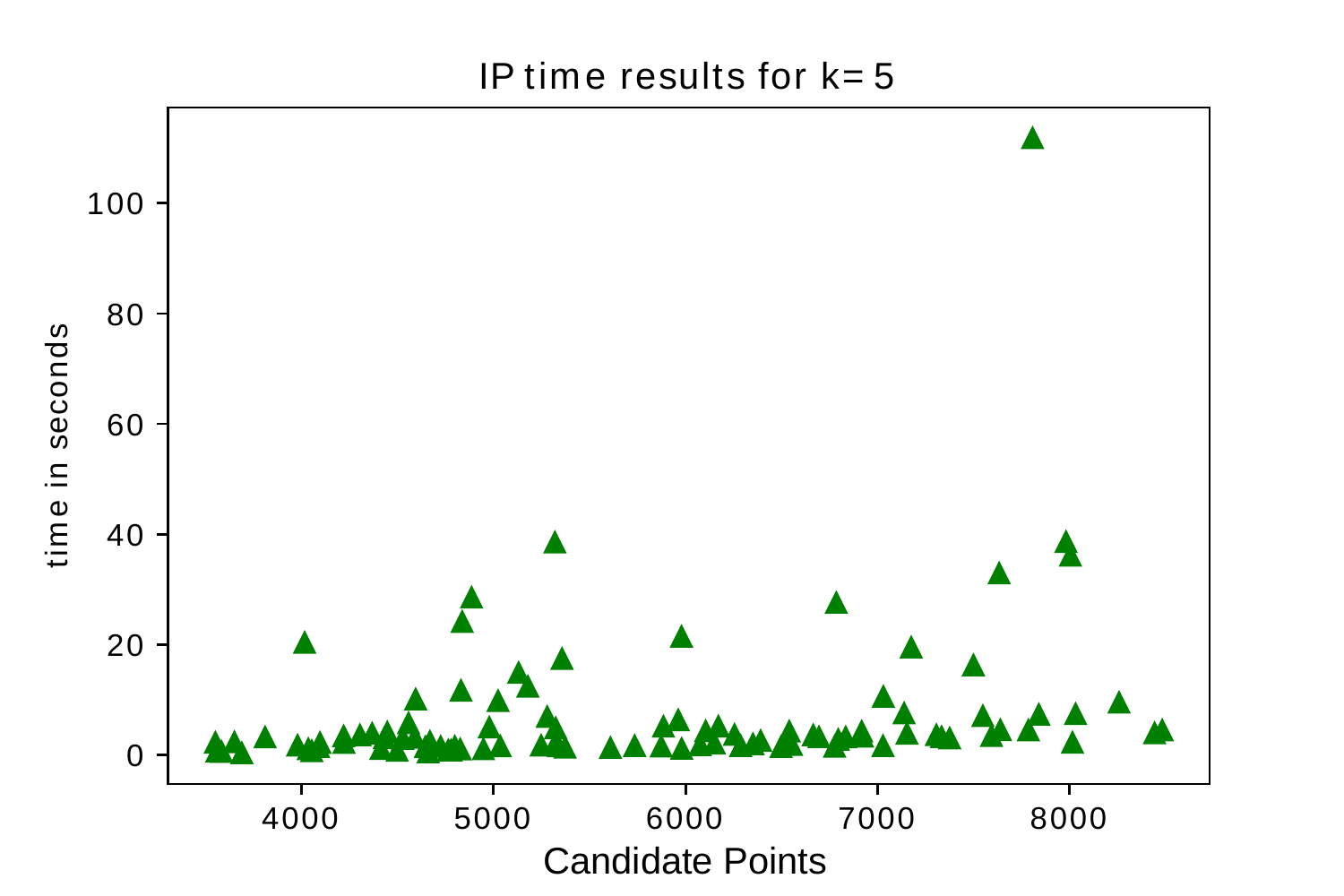}
    \subcaption{IP runtime; $k=5$.}\label{fig:TANOC_time5}
  \end{subfigure}
  \begin{subfigure}[t]{0.49\textwidth}
    \includegraphics[width=\textwidth]{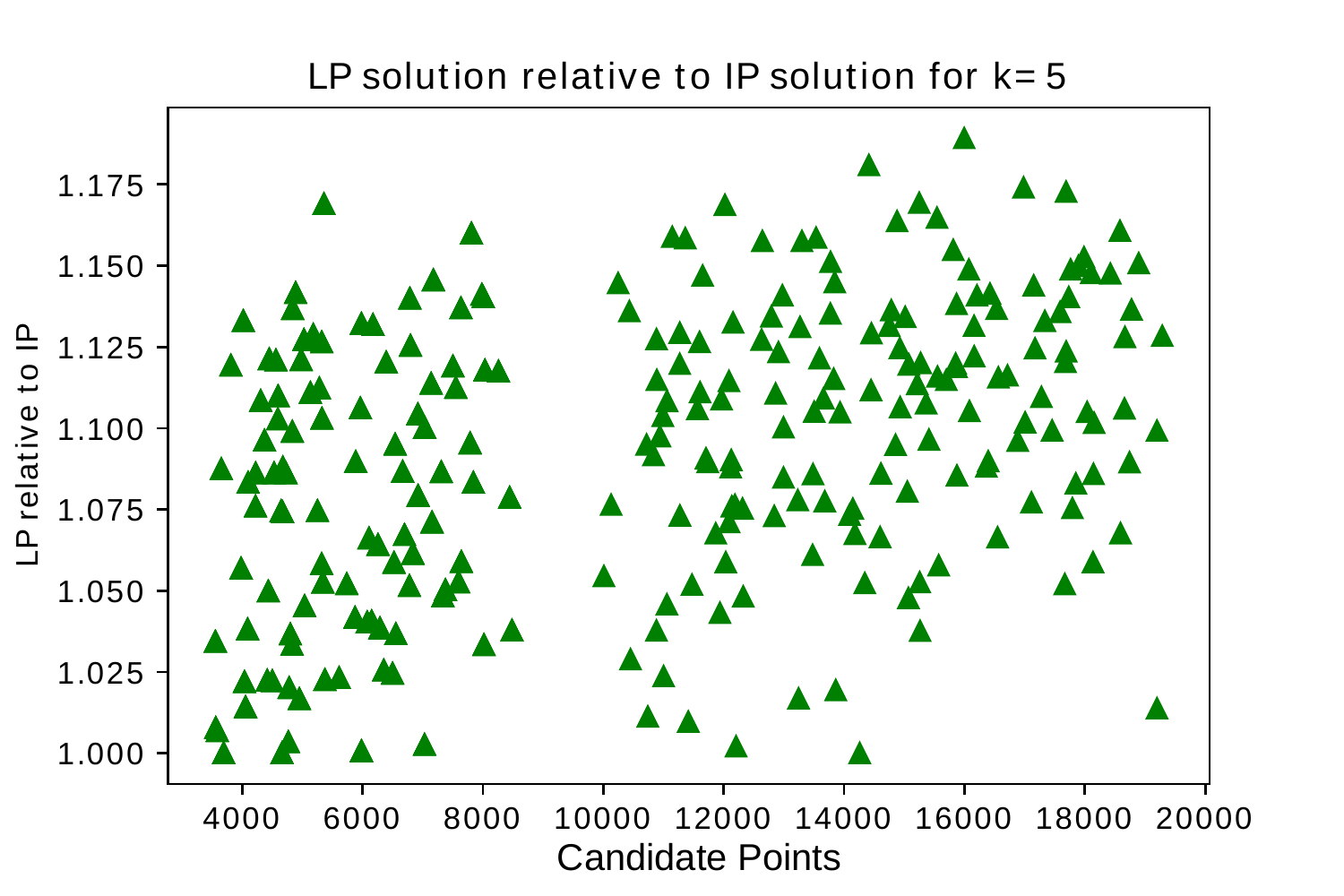}
    \subcaption{Integrality gap; $k=5$.}\label{fig:TANOC_gap5}
  \end{subfigure}
  \caption[Comparison solved algorithms]{IP runtime and integrality gap for axis-parallel instances and $k=5$.
  See Figure~\ref{fig:TANOC5_large} for larger images.}\label{fig:TANOC5}
\end{figure}

\begin{figure}
  \centering
  \begin{subfigure}[t]{0.49\textwidth}
    \includegraphics[width=\textwidth]{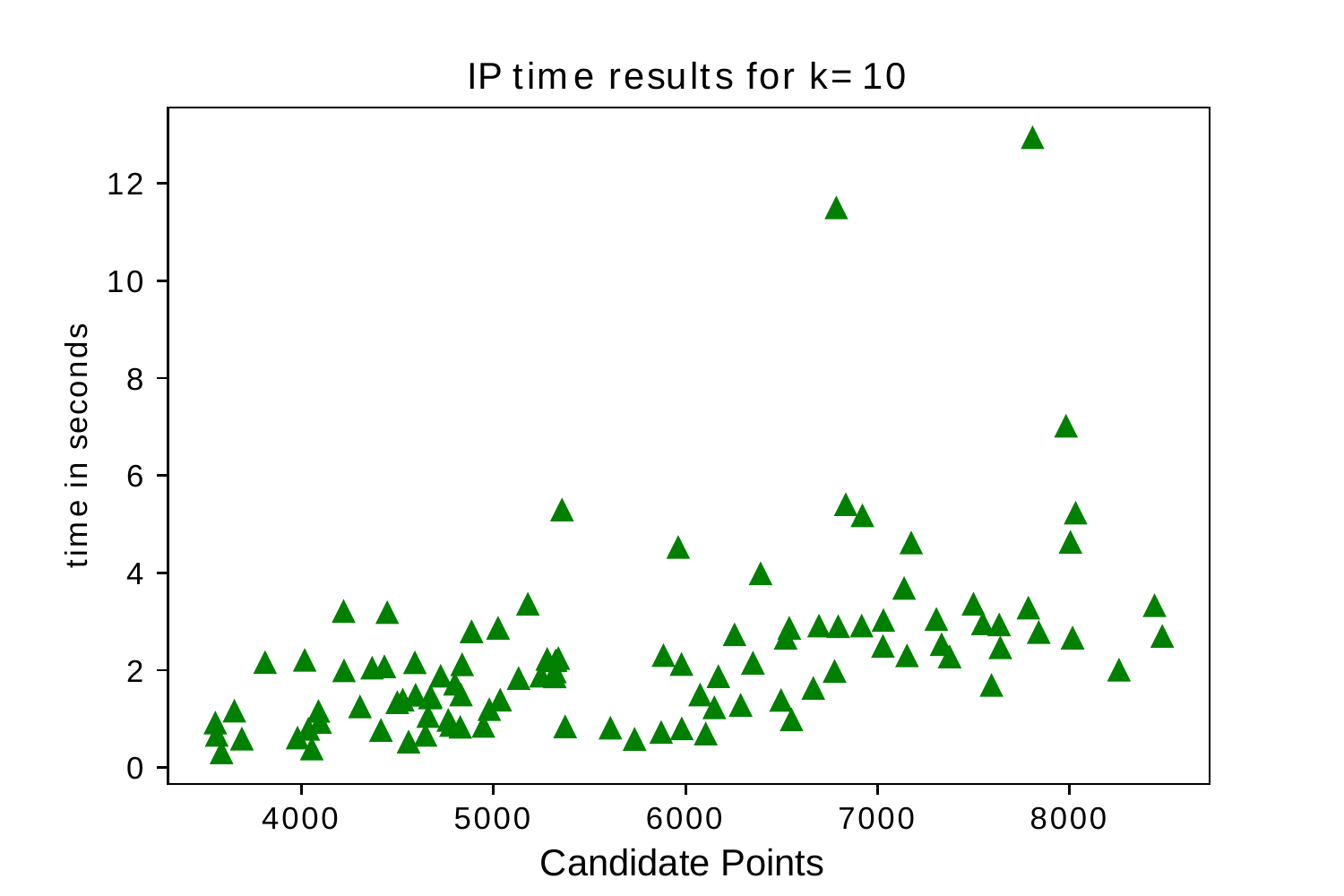}
    \subcaption{IP runtime; $k=10$.}\label{fig:TANOC_time10}
  \end{subfigure}
  \begin{subfigure}[t]{0.49\textwidth}
    \includegraphics[width=\textwidth]{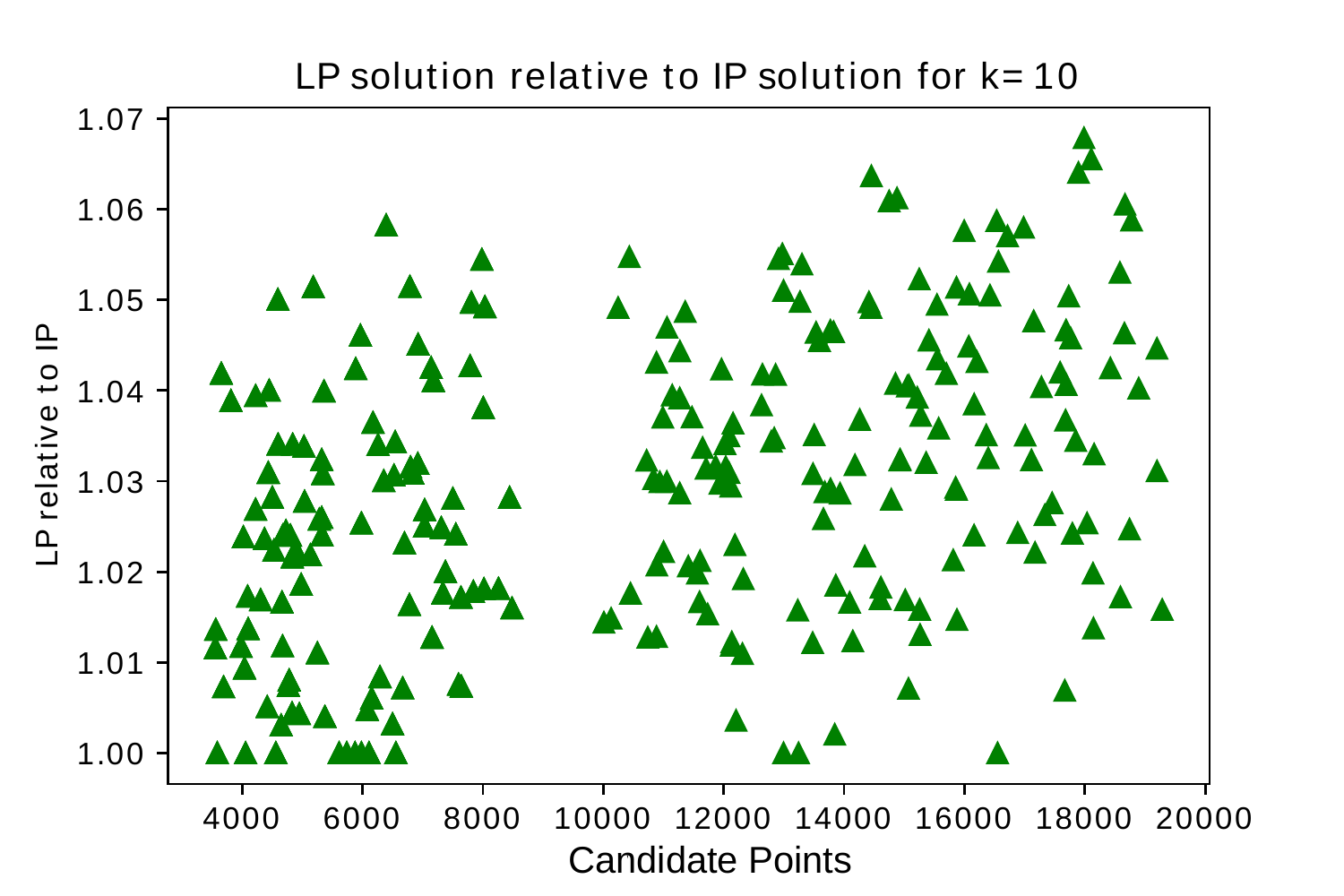}
    \subcaption{Integrality gap; $k=10$.}\label{fig:TANOC_gap10}
  \end{subfigure}
  \caption[Comparison solved algorithms]{IP runtime and integrality gap for axis-parallel instances and $k=10$.
  See Figure~\ref{fig:TANOC10_large} for larger images.}\label{fig:TANOC10}
\end{figure}

\begin{figure}
  \centering
  \begin{subfigure}[t]{0.49\textwidth}
    \includegraphics[width=\textwidth]{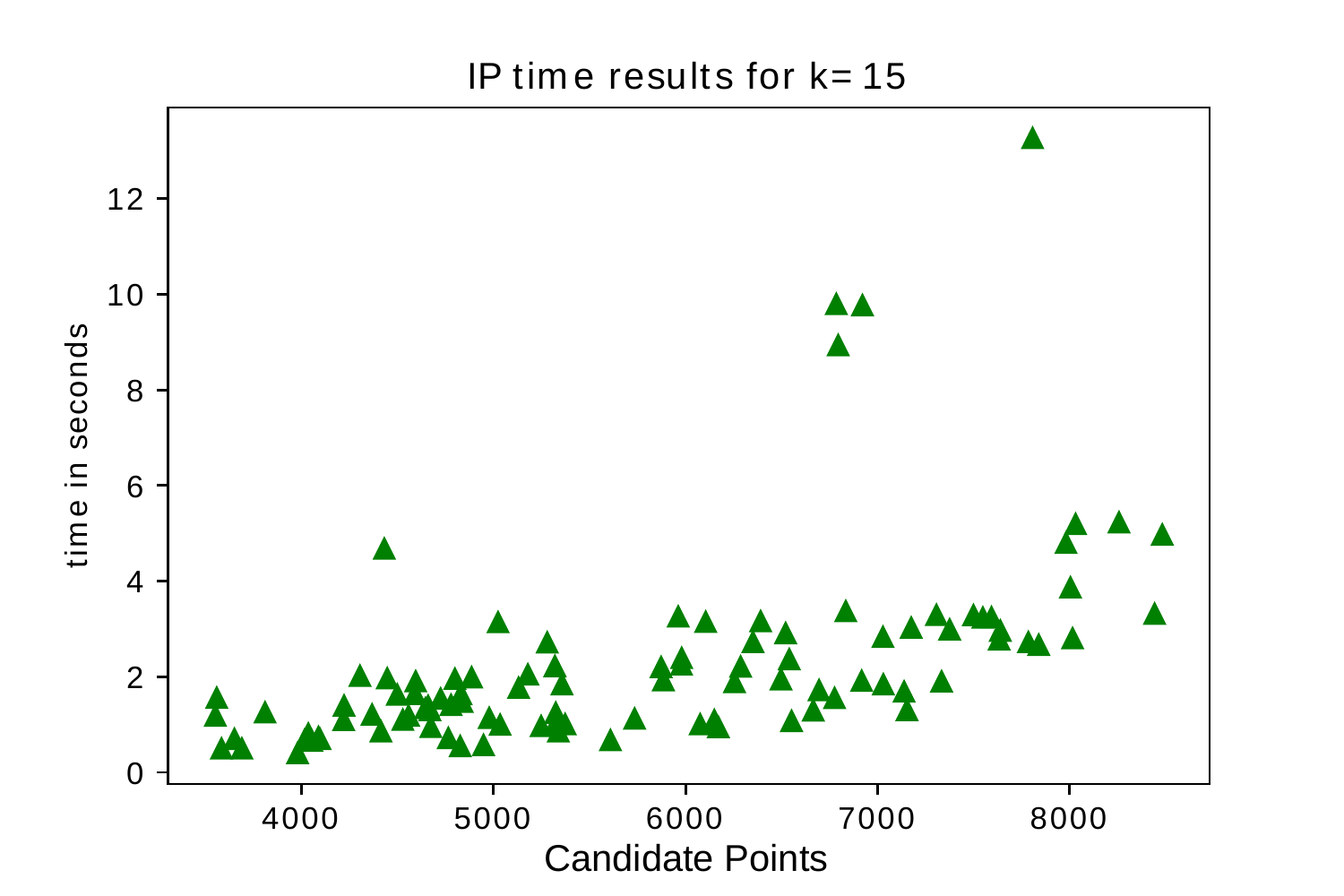}
    \subcaption{IP runtime; $k=15$.}\label{fig:TANOC_time15}
  \end{subfigure}
  \begin{subfigure}[t]{0.49\textwidth}
    \includegraphics[width=\textwidth]{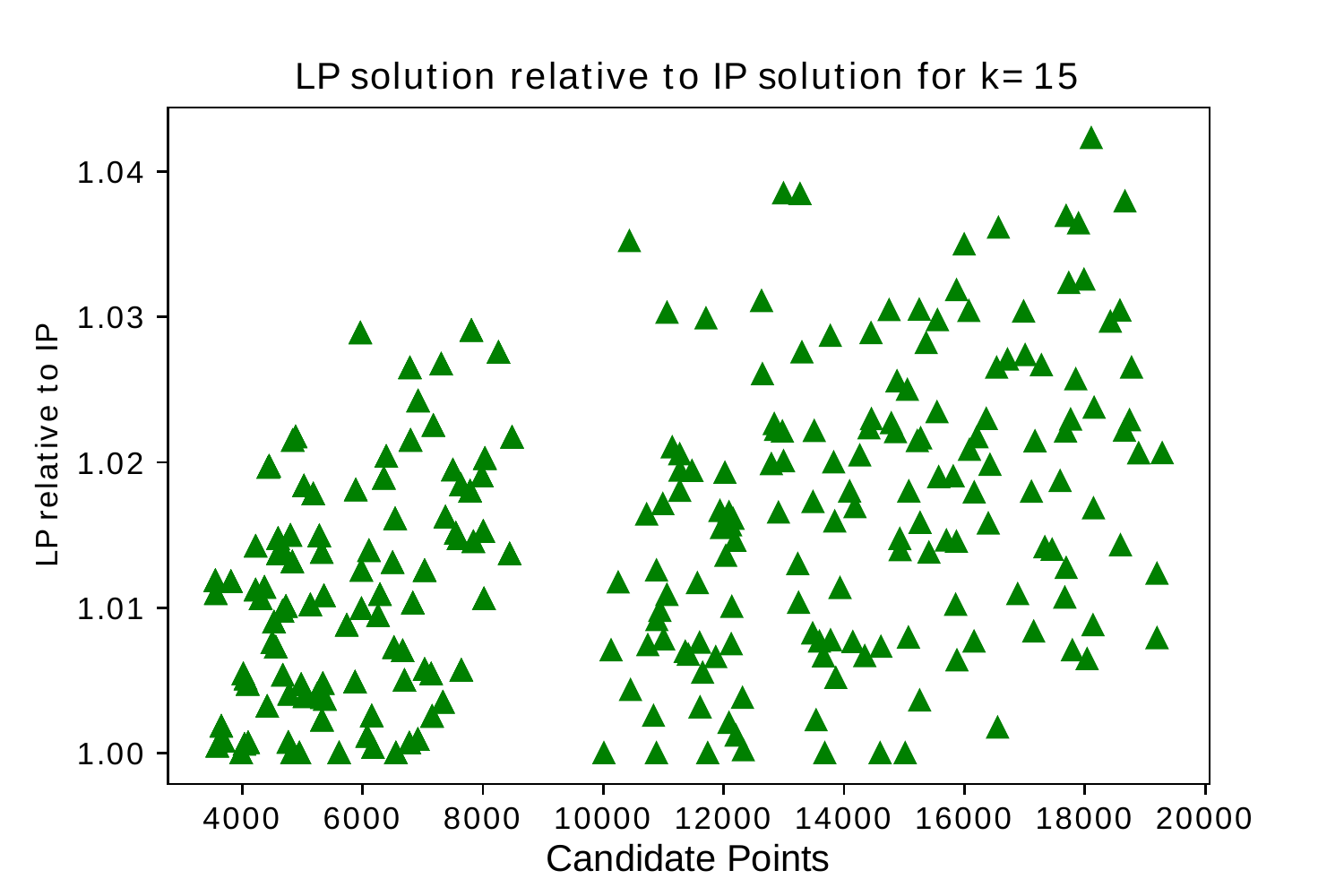}
    \subcaption{Integrality gap; $k=15$.}\label{fig:TANOC_gap15}
  \end{subfigure}
  \caption[Comparison solved algorithms]{IP runtime and integrality gap for axis-parallel instances and $k=15$.
  See Figure~\ref{fig:TANOC15_large} for larger images.}\label{fig:TANOC15}
\end{figure}

\begin{figure}
  \centering
  \begin{subfigure}[t]{0.49\textwidth}
    \includegraphics[width=\textwidth]{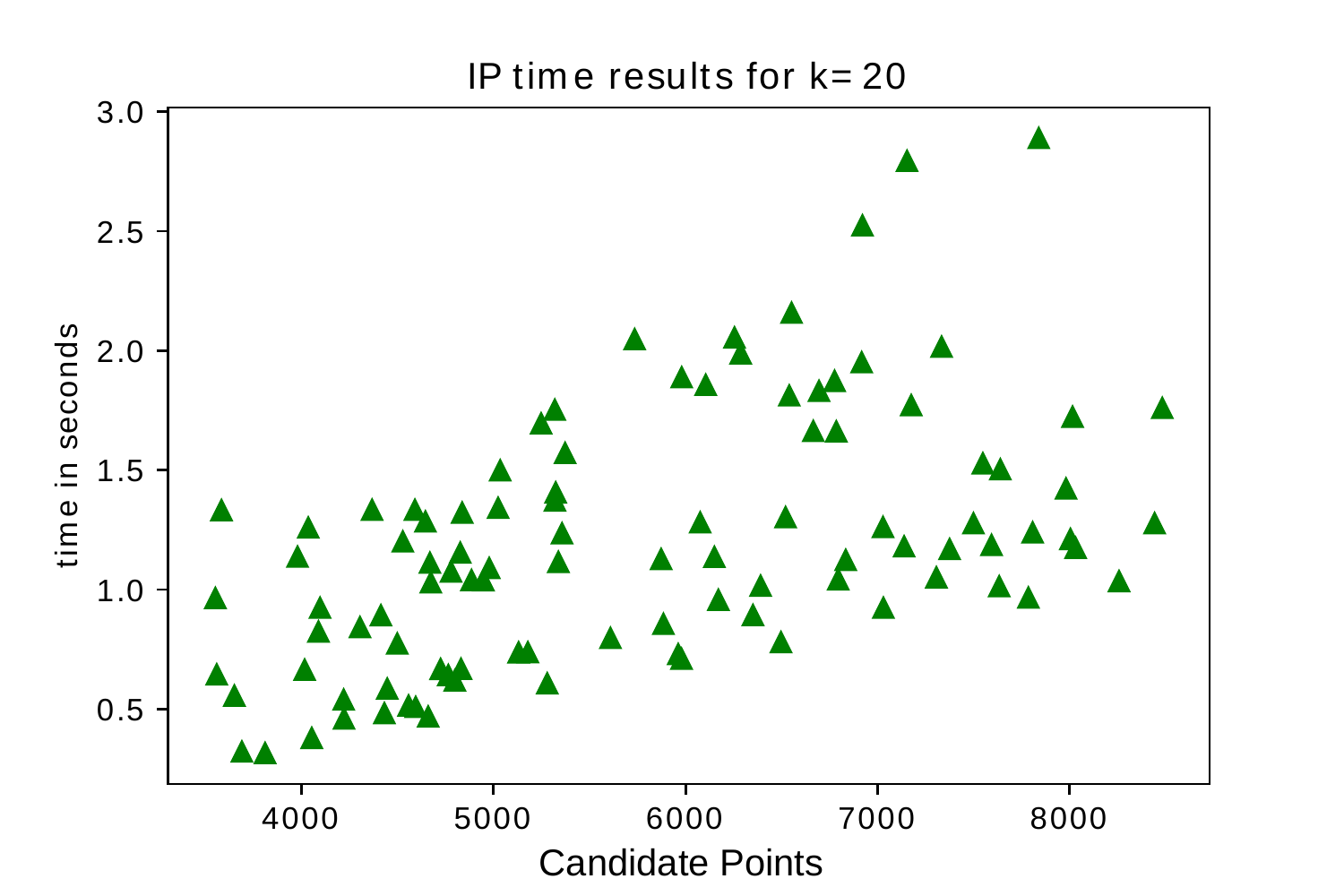}
    \subcaption{IP runtime; $k=20$.}\label{fig:TANOC_time20}
  \end{subfigure}
  \begin{subfigure}[t]{0.49\textwidth}
    \includegraphics[width=\textwidth]{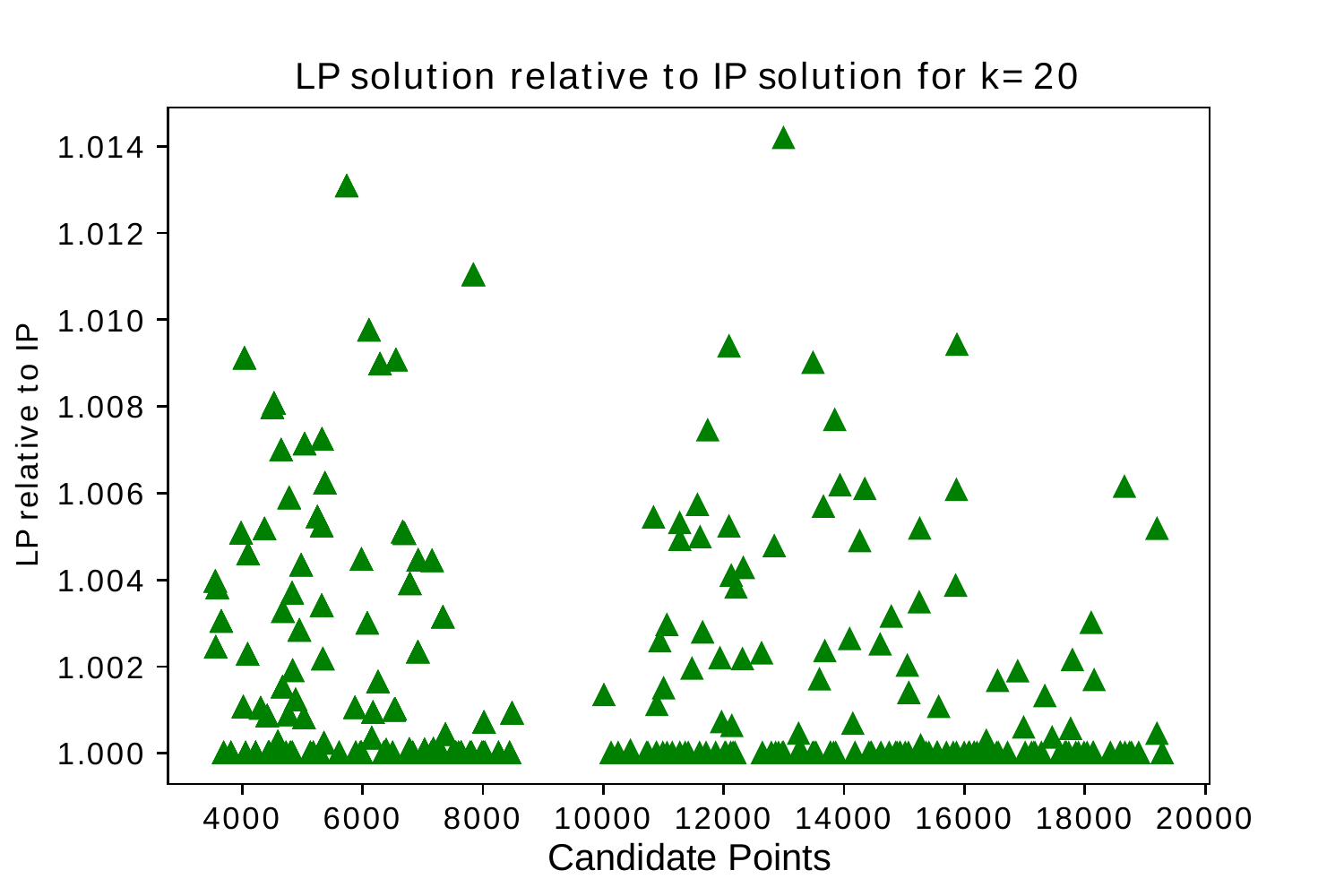}
    \subcaption{Integrality gap; $k=20$.}\label{fig:TANOC_gap20}
  \end{subfigure}
  \caption[Comparison solved algorithms]{IP runtime and integrality gap for axis-parallel instances and $k=20$.
  See Figure~\ref{fig:TANOC20_large} for larger images.}\label{fig:TANOC20}
\end{figure}

\begin{figure}
  \centering
  \begin{subfigure}[t]{0.49\textwidth}
    \includegraphics[width=\textwidth]{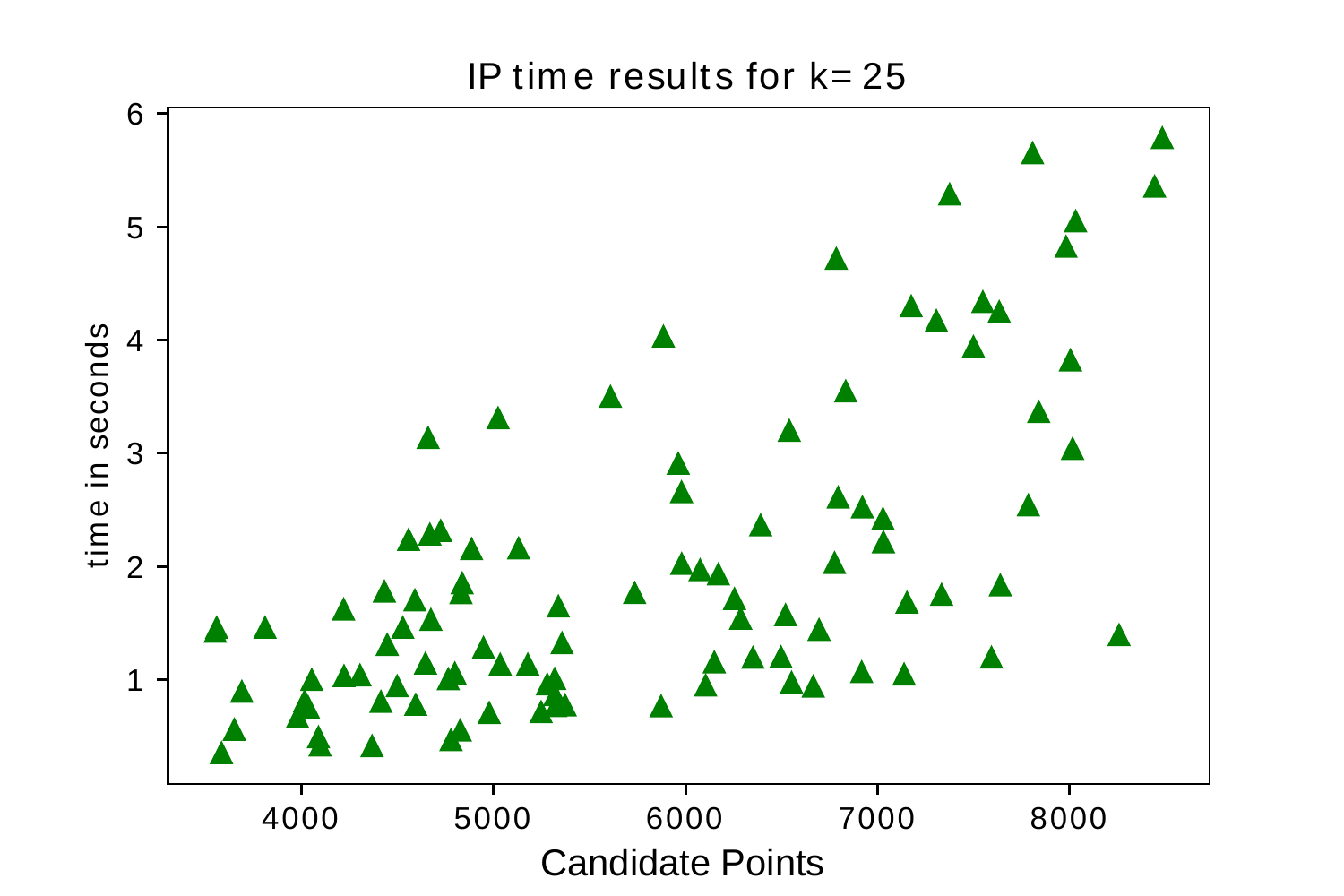}
    \subcaption{IP runtime; $k=25$.}\label{fig:TANOC_time25}
  \end{subfigure}
  \begin{subfigure}[t]{0.49\textwidth}
    \includegraphics[width=\textwidth]{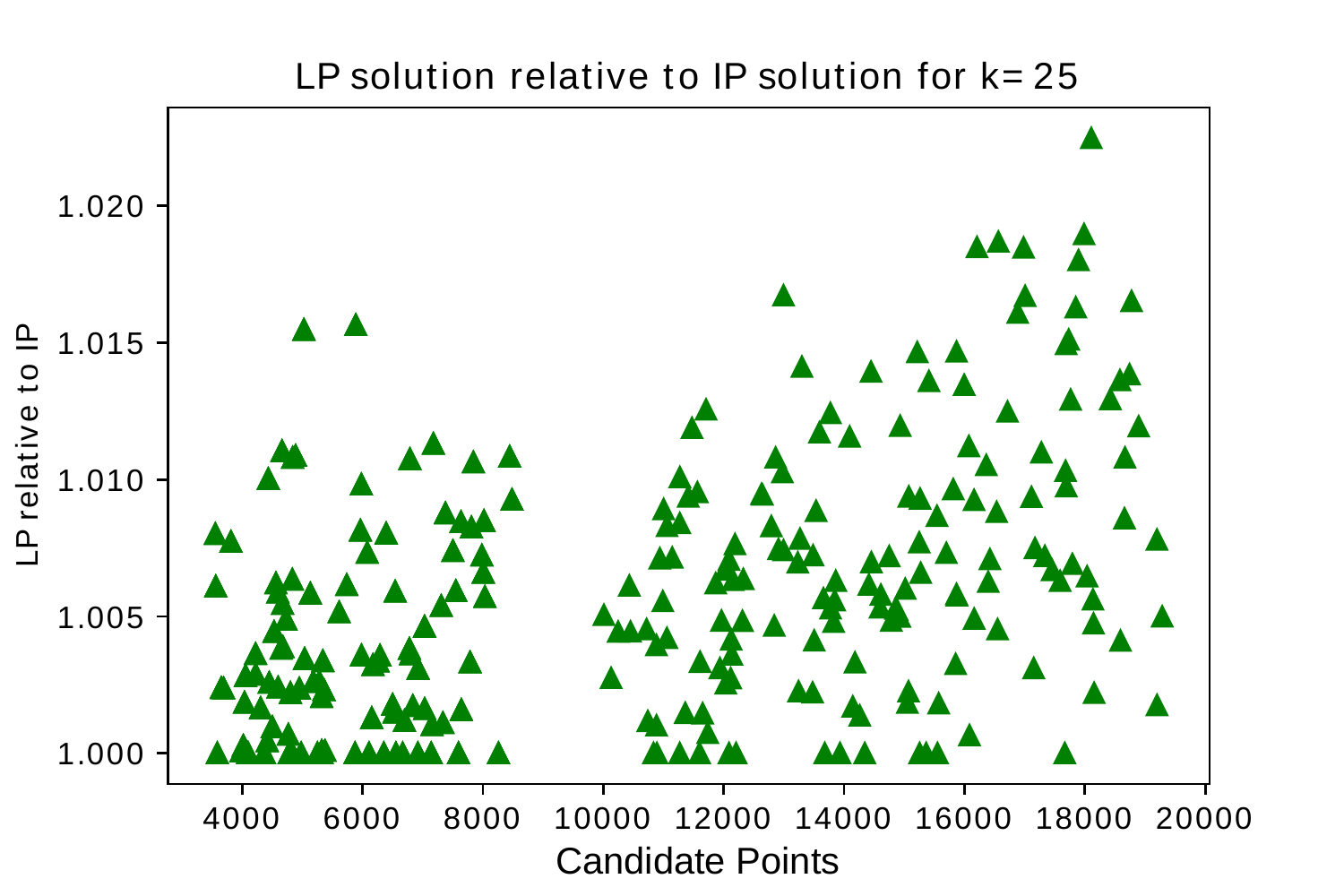}
    \subcaption{Integrality gap; $k=25$.}\label{fig:TANOC_gap25}
  \end{subfigure}
  \caption[Comparison solved algorithms]{IP runtime and integrality gap for axis-parallel instances and $k=25$.
  See Figure~\ref{fig:TANOC25_large} for larger images.}\label{fig:TANOC25}
\end{figure}
}

\subsubsection{Application to Taxi Trajectory Data}  

We have applied our TCP model to solve real-world data sets to optimality.
In Figure~\ref{fig:taxi_opt} we show the results of computing
$k=5$ optimal portals for a set of trajectories based on taxi cab
routes in the San Francisco Bay Area. The data is based on 375
vehicles, sampled every 5 minutes, 288 times per day, for one
week~\cite{epfl-mobility-20090224}; see the trajectories in Figure~\ref{fig:sanfran}.

Our experiments included runs on 30 instances, with $k$ ranging from 5
to 11, on sets of 10 to 120 trajectories of varying lengths (comprised
of 1300 to 3700 edges, and 600 to 1800 vertices).  The trajectories
are snapped to a regular grid graph.  Solution times of the IP were up to 200
seconds of computation, with most instances taking less than 10
seconds.

\begin{figure}[h]
  \centering
  \includegraphics[width=0.4438\textwidth]{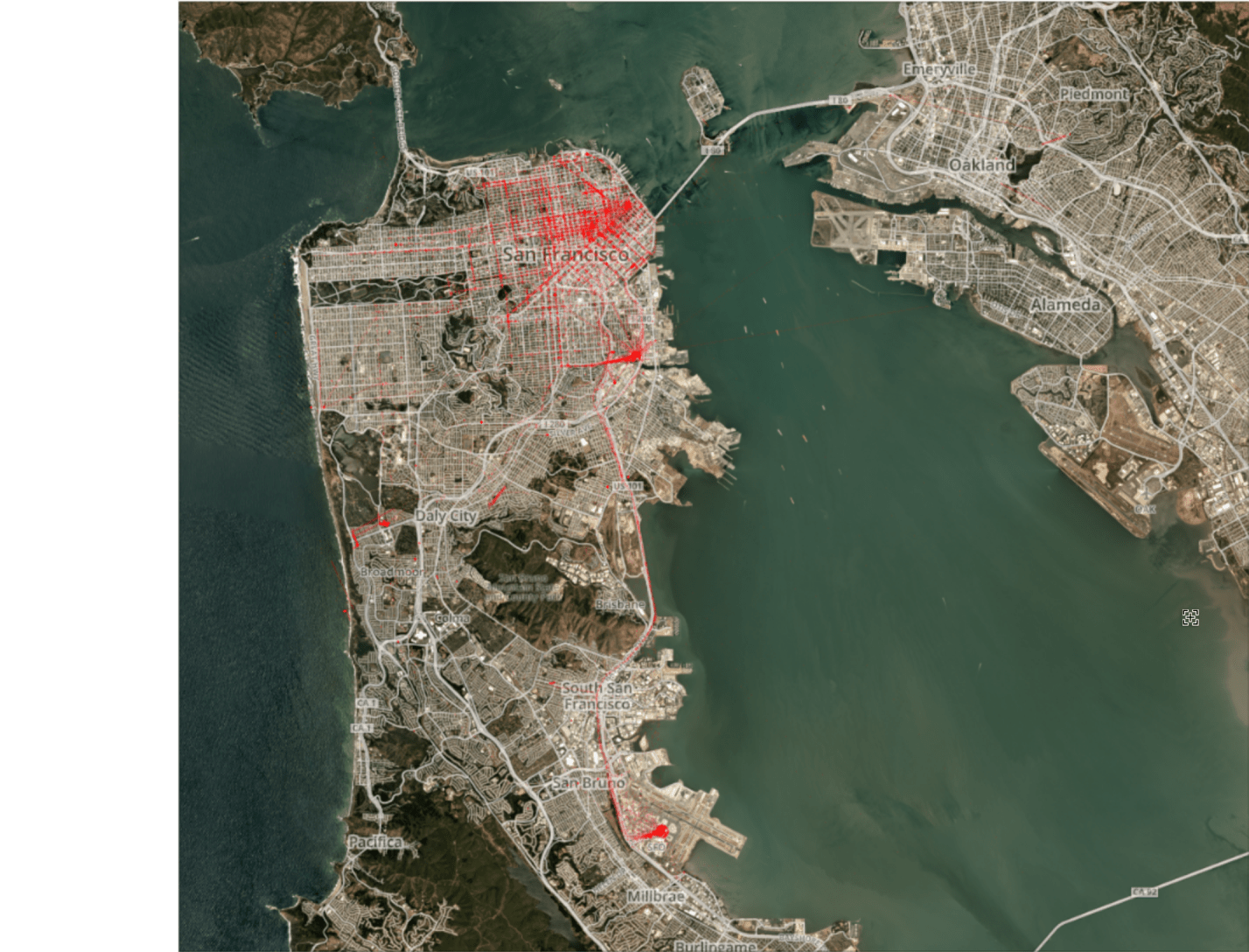} 
\hfill
  \includegraphics[width = 0.36\textwidth]{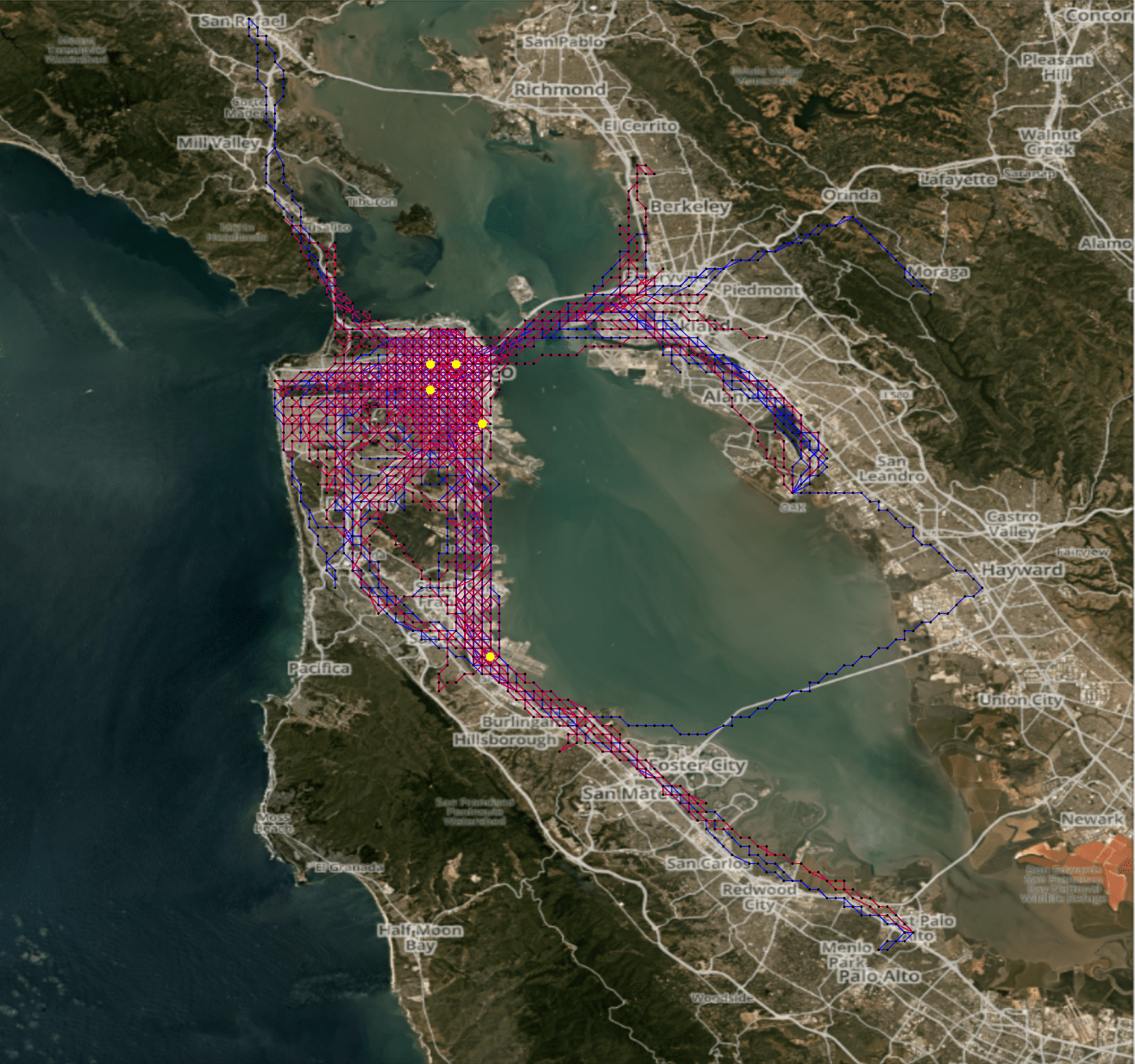}	 
\hfill
  \caption{(Left) Candidate points before processing. (Right) Solution to real-world TCP instance: An optimal set of $k=5$ portals are highlighted. Red trajectory portions (some of which may loop back) are captured, blue ones are not captured.
Satellite images are courtesy of Planet Labs Inc.}
  \label{fig:taxi_opt}
\vspace*{-0.3cm}
\end{figure}


\section{Conclusion}

We have introduced the trajectory capture problem (TCP), an
optimization problem in which we seek to place $k$ points (or portals)
in order to ``capture'' the maximum total length of a given set of
paths/trajectories between two placed points. We have shown that the
problem is NP-hard, even for axis-aligned line segment trajectories in
the plane, and we have given approximation algorithms for two cases.
Can we improve the approximation factor of $K$ for a set of trajectories
that is the union of $K$ subsets, each of which is noncrossing?
Can we improve the approximation factor of $\Delta$ for a set of trajectories
of depth at most $\Delta$?

A focus of our work is the exploration, via algorithm
engineering, of practical methods for solving the TCP. Our methods are
based on integer programming and on simple heuristic search methods.
It will be interesting to develop more specified methods for other,
specific classes of instances, such as 
further geometric instances arising from other types of real-world geographic data.

\section*{Acknowledgments}
Work by Tyler Mayer was mostly carried out while at Stony Brook University.
Joe Mitchell and Tyler Mayer were partially supported by the National Science
Foundation (CCF-1526406) and a grant from the US-Israel Binational Science
Foundation (BSF project 2016116). Joe Mitchell was also partially supported by
the DARPA Lagrange program. Sandia National Laboratories is a multimission
laboratory managed and operated by National Technology and Engineering
Solutions of Sandia, LLC, a wholly owned subsidiary of Honeywell International,
Inc., for the U.S. Department of Energy's National Nuclear Security
Administration under contract DE-NA-0003525.


\bibliography{refs}

\setcounter{section}{0}
\renewcommand\thesection{\Alph{section}}

\section{Appendix}
\label{app:eval}

\subsection{Proof of Theorem~\ref{thm:general-hard}}

\begin{proof}
  The reduction is from the {\sc Hitting Lines} problem: Decide if
  there exists a set of $k$ points that hit every one of a given set
  of $n$ lines in the plane.

  Consider an instance, ${\cal L}=\{\ell_1,\ldots,\ell_n\}$, of $n$
  lines in the plane.  Let $\delta>0$ be the radius of a disk,
  $D_\delta$, centered at the origin $(0,0)$, that is large enough to
  contain all intersection points (crossing points) in the arrangement
  of ${\cal L}$; it suffices to select $\delta$ greater than the
  Euclidean distance from the origin to any crossing point, $\ell_i
  \cap \ell_j$.  Let $R>>\delta$ and consider the much larger disk,
  $D_{R+\delta}$, centered at the origin; it will suffice to pick $R=
  2n\delta$.  Each line $\ell_i$ intersects $D_\delta$ in a single
  segment, $s_i$, and intersects the annulus $D_{R+\delta}\setminus D_\delta$
  in two segments, $s'_i$ and $s''_i$, with $s'_i$ in the halfspace
  $\{(x,y)\in \Re^2: y\leq 0\}$ (and $s''_i$ in the halfspace with
  $y\geq 0$).

  We consider the instance of TCP with the following $2n$ line
  segments as input: the $n$ segments, $s_i\cup s'_i$, obtained by
  concatenating $s_i$ and $s'_i$, and the $n$ segments $\sigma_i$ that
  connect the endpoint of $s'_i$ (the one at distance $R+\delta$ from
  the origin) to a point $p$, at distance $R'+R>>R$ from the origin;
  it will suffice to pick $R'=2nR$.  Refer to
  Figure~\ref{fig:n-orientation-hard}.
  
\begin{figure}[htb]
\centering
\includegraphics[width = 0.4\textwidth]{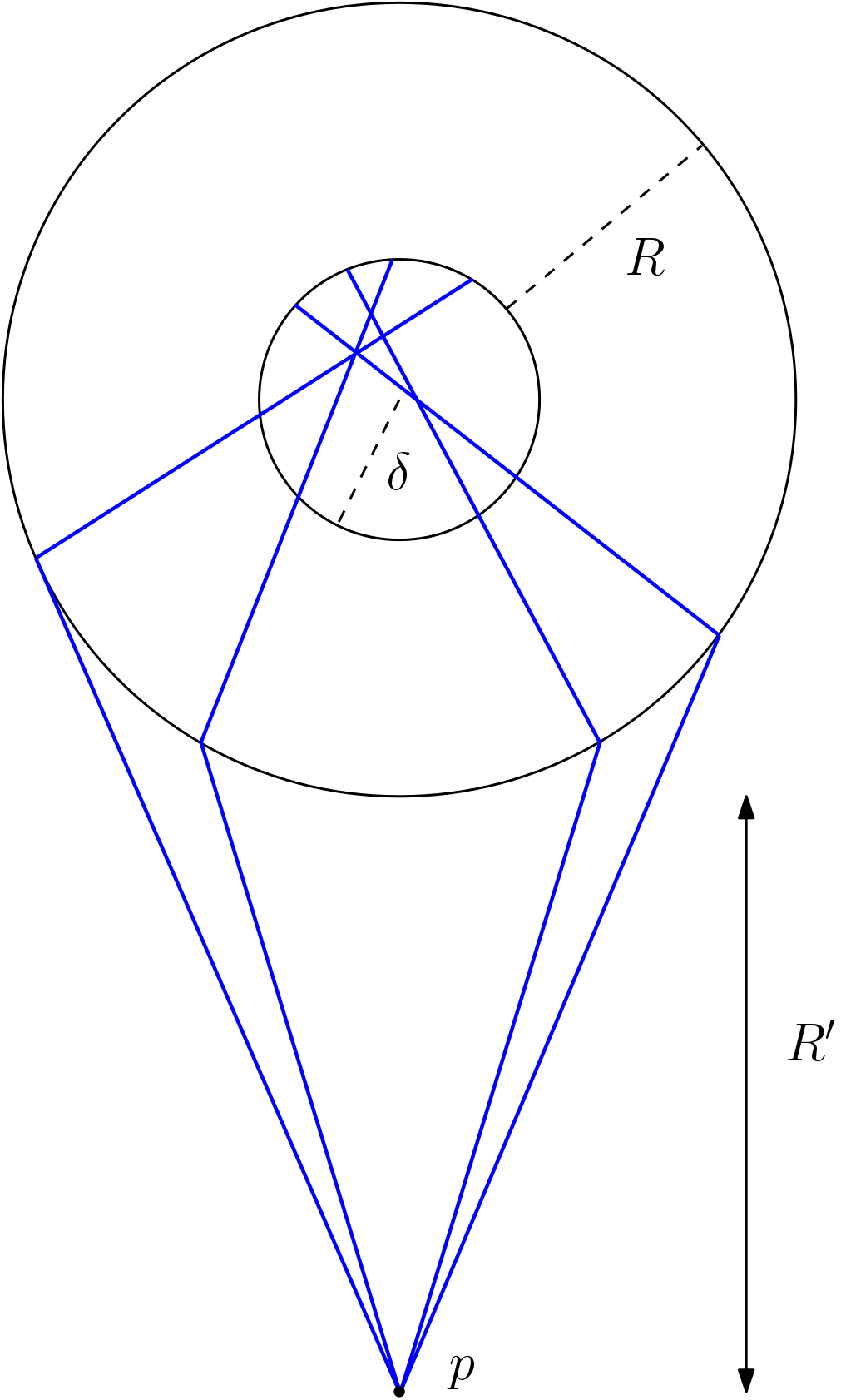}
\caption{Reduction from the problem {\sc Hitting Lines}.}
\label{fig:n-orientation-hard}
\end{figure}

  We claim that there is a hitting set of size $k$ for ${\cal L}$ if
  and only if it is possible to capture at least length $\sum_i
  (|s'_i| + |\sigma_i|)$ using a budget of $n + 1 + k$ points in the
  instance of TCP.

  First, if there exists a hitting set of size $k$ for ${\cal L}$,
  then in the instance of TCP, we place $k$ hit points (within
  $D_\delta$) to hit all of the segments $s_i\cup s'_i$, and place $n$
  hit points at the endpoints of these $n$ segments that lie on the
  boundary of $D_{R+\delta}$, and one point at $p$; this suffices to
  capture at least length $\sum_i (|s'_i| + |\sigma_i|)$.

  For the converse, the only way to capture length at least $\sum_i
  (|s'_i| + |\sigma_i|)$ is to place hit points at $p$, at the $n$
  points $\sigma_i\cap s'_i$, and at $k$ points within $D_\delta$ that
  hit all $n$ of the segments $s'_i$.  (Lengths captured within
  $D_\delta$ are very small compared to the lengths captured outside
  of $D_\delta$.)
\end{proof}

\subsection{Proof of Theorem~\ref{thm:2orientation-hard}}

\begin{proof}
  Without loss of generality, we assume that the two orientations are
  horizontal and vertical; i.e., the segments of ${\cal T}$ are
  axis-parallel.
  
The reduction is from 3-SAT, in which one is to decide if there is a
truth assignment for $n$ Boolean variables $x_i : 1 \leq i \leq n$
that satisfy a set of $m$ clauses in conjunctive normal form.  The
construction is similar to that used to show NP-hardness of the
hitting set problem on axis-parallel
segments~\cite{fekete2018geometric}, with some key new features.
See Figure~\ref{fig:THM3} for the overall construction.

\begin{figure}[htbp]
\centering
\includegraphics[width=0.6\textwidth]{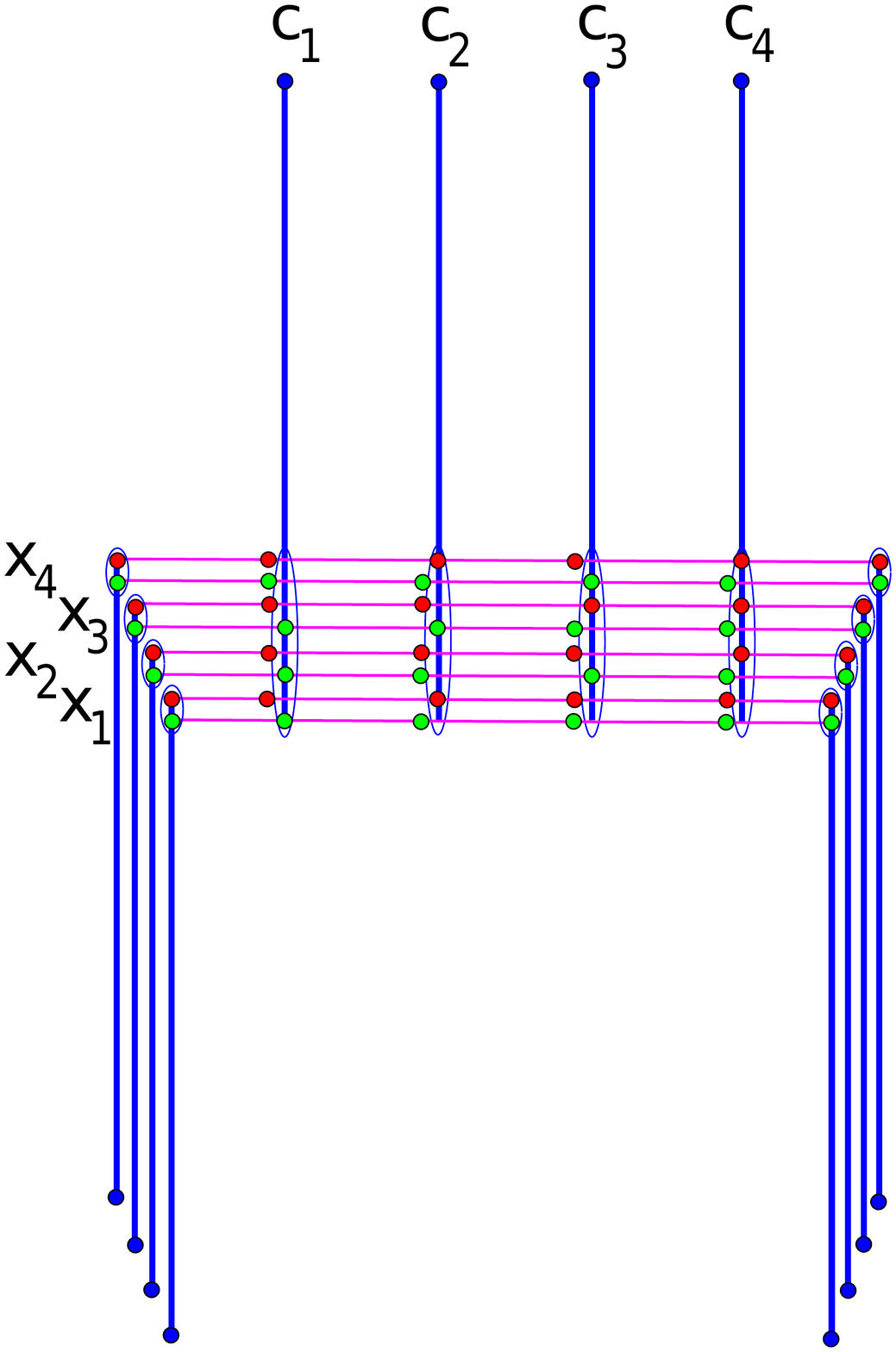}
\caption{Overview of the hardness construction, for the 3SAT instance
$(x_1\vee x_2 \vee x_3)\wedge(\overline{x}_1\vee x_3\vee\overline{x}_4)
\wedge(x_2\vee\overline{x}_3\vee x_4)\wedge (\overline{x}_2\vee \overline{x}_3\vee\overline{x}_4)$.
For an instance with $m$ clauses and $n$ variables, it consists of $2n+m$ vertical (blue) segments,
each of length $nm$, and $2n(m+1)$ horizontal (purple) segments of unit length (up to small modifications of 
size $O(\varepsilon)$). The red and green dots indicate points at which different
segments meet. With $4n+m+nm$ portals, a total distance of length at least $n(m+1)+2n^2m+nm^2-\Theta(nm\varepsilon)$
can be captured if and only if the 3SAT instance has a satisfying truth instance.
  }
\label{fig:THM3}
\end{figure}

First, we place $m$ vertical \emph{clause} segments, each of length $nm$,
evenly spaced at unit distance; these are shown in blue at the top 
of Figure~\ref{fig:THM3}. We also place $2n$ vertical
\emph{variable} 
segments of length $nm$ (also shown in blue) at unit distance to the right and left of the clause
segments, such that their top vertices line up with the bottom vertices
of the clause segments; these variable segments are slightly shifted vertically and horizontally
by $\Theta(\varepsilon)$ for a sufficiently small $\varepsilon$ (it suffices to set $\varepsilon=\Theta(1/mn)$).

For each of the $n$ variables, we add a set of 2 \emph{literal} horizontal chains of (approximately) unit segments
(shown in purple in the figure); each chain consists of $m+1$ such segments,
and two consecutive segments in the same chain intersect only in their shared endpoint;
these are shown by red and green dots in the figure.
Similarly, the first and last endpoint of such chains of horizontal segments lie on the
corresponding left and right vertical variable segments. The vertical distance between
the two chains is $\Theta(\varepsilon)$.
For each variable, the lower chain corresponds to a {\tt true} 
assignment (corresponding to green dots), while the upper chain corresponds to a 
{\tt false} assignment of that variable (corresponding to red dots).

Finally, the lengths of the horizontal segments for a literal are adjusted by $\Theta(\varepsilon)$,
so that precisely the dots corresponding to the literals in a particular clause are
positioned on the vertical segments for that clause.

Now we claim: With $4n+m+nm$ portals, a total distance of length at least $n(m+1)+2n^2m+nm^2-\frac{1}{2}$
can be captured, if and only if the 3SAT instance has a satisfying truth instance.

It is straightforward to see that a satisfying truth assignment induces a set of 
$4n+m+nm$ portals that capture a total distance of length at least $n(m+1)+2n^2m+nm^2-\Theta(nm\varepsilon)$:
Simply pick the upper endpoints of vertical clause segments and the bottom endpoints of vertical
variable segments, along with all of the endpoints of horizontal literal segments for the chain
corresponding to the truth assignments of a variable. This captures the full distance of 
$m+1-\Theta(n\varepsilon)$ of one horizontal chain of literal segments for each of the $n$ variables,
$nm-\Theta(n\varepsilon)$ of each of the $2n$ vertical variable segments, as well as 
$nm-\Theta(n\varepsilon)$ of each of the $m$ vertical clause segments. 

For the converse, first observe that any solution satisfying the bound must capture the full length of all
vertical segments, up to a total difference of at most $\frac{1}{2}$. As a consequence, we can assume
that $2n+m$ portals must be placed at endpoints of the vertical segments, as indicated by the 
$2n+m$ blue dots in the figure. Furthermore, each vertical segment must have a second portal 
near its other endpoint; without decreasing the value of the solution
by more than $\Theta(n\varepsilon)$, we can assume that each such portal is placed 
on one of the endpoints of a horizontal literal segment. This captures a total distance
of $2n^2m+nm^2-\Theta(nm\varepsilon)$ from the vertical segments.

On the other hand, this leaves at most $2n+nm=n(m+2)$ (non-blue) portals to 
capture a distance of at least $n(m+1)-\frac{1}{2}$ from the horizontal unit-length
segments. If we consider the auxiliary graph in which these portals
are represented by vertices, and two vertices are adjacent if they 
capture the same horizontal edge, we conclude that this graph decomposes
into a set of paths; furthermore, a path consisting of $p+1$ vertices
and $p$ edges captures a distance of at most $p$. As a consequence,
there can be at most $n$ such paths, otherwise we capture
a horizontal distance of at most $n(m+1)-1$. Because the longest possible length
of a path is $m+1$, which occurs if and only if we pick all endpoints
in a full horizontal literal chain, it follows that we must pick
precisely $n$ such full chains, with either of their two endpoints on 
the two vertical segments for the corresponding variable. As this leaves
$2n$ portals to be positioned near the upper endpoints of the $2n$ vertical variable
segments, and we need at least one portal on each of the vertical variable segments,
we conclude that for each variable we must pick precisely one literal chain of
endpoints, corresponding to a truth assignment. Finally, the choice of portals
must also pick at least one portal near the bottom end of each vertical variable segment;
this is only possible if the choice of literal chains produces at least
one satisfying literal for each clause. This concludes the claim.
\end{proof}

\subsection{Proof of Theorem~\ref{thm:gap}}

\begin{proof}
Consider the family of examples that arise from the arrangement of segments corresponding to the embedding of the complete graph, $K_n$, with its $n$ nodes embedded as evenly spaced points, $U$, on the boundary of a circle of diameter 1; see Figure~\ref{fig:intProgCircleGadget}.

The corresponding arrangement graph ${\cal G}$ has vertices at the $n$ points $U$, as well as at the vertices (crossing points) in the arrangement; edges of ${\cal G}$ connect consecutive vertices along the segments.
The set ${\cal T}$ includes a trajectory (of collinear edges) along each of the ${n \choose 2}$ line segments.

\begin{figure}
\centering
\includegraphics[scale=0.8]{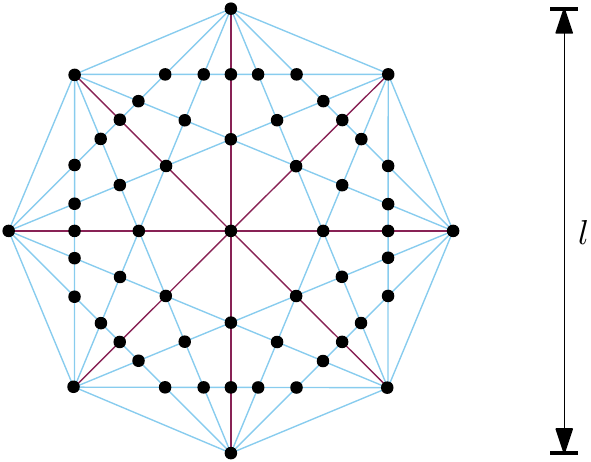}
\caption{Example used in the proof of Theorem~\ref{thm:gap}.}
\label{fig:intProgCircleGadget}
\end{figure}

For $k=1$, the integrality gap is infinite, because the integral solution
cannot capture anything, while the fractional solution can capture the edges
fractionally.  For $k>1$, the optimal solution can capture less than $k^2$
trajectories, with a maximal trajectory length of $1$, so $OPT_{INT}(k)\leq
k^2$.
The fractional solution, however, can include fractional portals, with variable values $k/n$, at each of the $n$ points $U$, resulting in each trajectory being captured fractionally with value $k/n$.
For simplicity and without loss of generality, assume that $n$ is a multiple of $4$.
Then, each point $U$ has at least $n/2$ incident segments with length at least $0.5$.
  To see this, simply look onto the left most vertex: Each point of $U$ on the right half has a distance of at least $0.5$ (in fact, at least $\sqrt{0.5}$).
This means that there are at least $n\cdot \frac{n}{4}$ trajectories with length of at least $0.5$, so the overall length of all trajectories is at least $\frac{n^2}{8}$.
These trajectories are all captured with fraction at least $1/k$, so $OPT_{FRAC}(k)\geq \frac{1}{k}\cdot \frac{n^2}{8}$.
With 
\[\frac{OPT_{FRAC}(k)}{OPT_{INT}(k)} \geq \frac{\frac{1}{k}\cdot \frac{n^2}{8}}{k^2} \geq \frac{n^2}{8k^3},\]
we obtain an arbitrarily large integrality gap, as $n$ increases (for any fixed $k$).
\end{proof}
  
  \subsection{Local vs.~Global Neighborhoods}
  \label{sec:local-global}

A test for the comparison between Integer Linear Programming and
heuristic approaches in Section~\ref{IPvsHeur} focuses on choosing
good neighborhoods for the meta-heuristics.  To this end, we compared
the results of local vs.\ global neighborhoods with respect to
quality and time.

\begin{figure}
\centering
\begin{subfigure}[t]{0.49\textwidth}
\includegraphics[width=\textwidth]{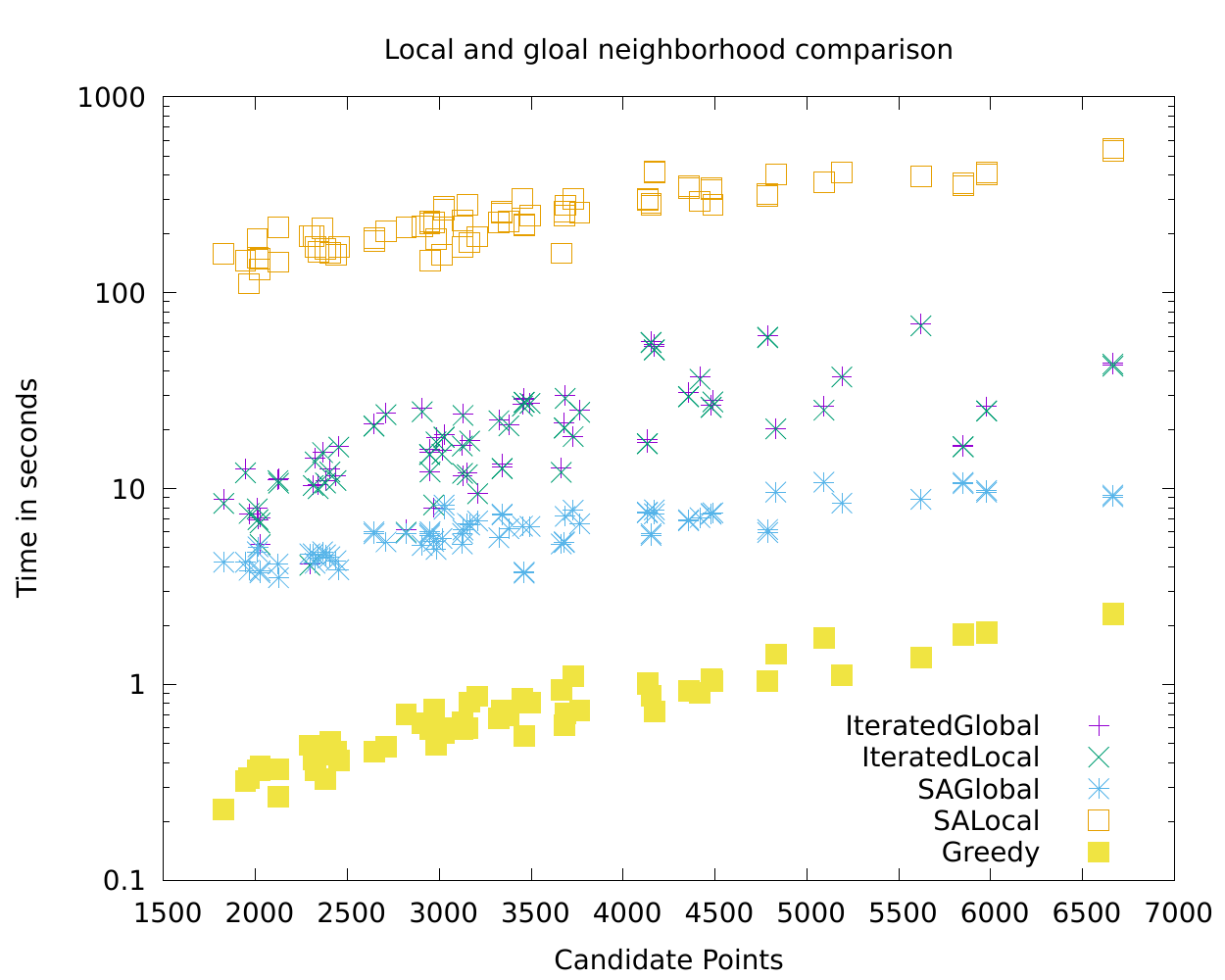}
\subcaption{Comparison of required time}
\label{fig:NeighCompTime}
\end{subfigure}
\begin{subfigure}[t]{0.49\textwidth}
\includegraphics[width=\textwidth]{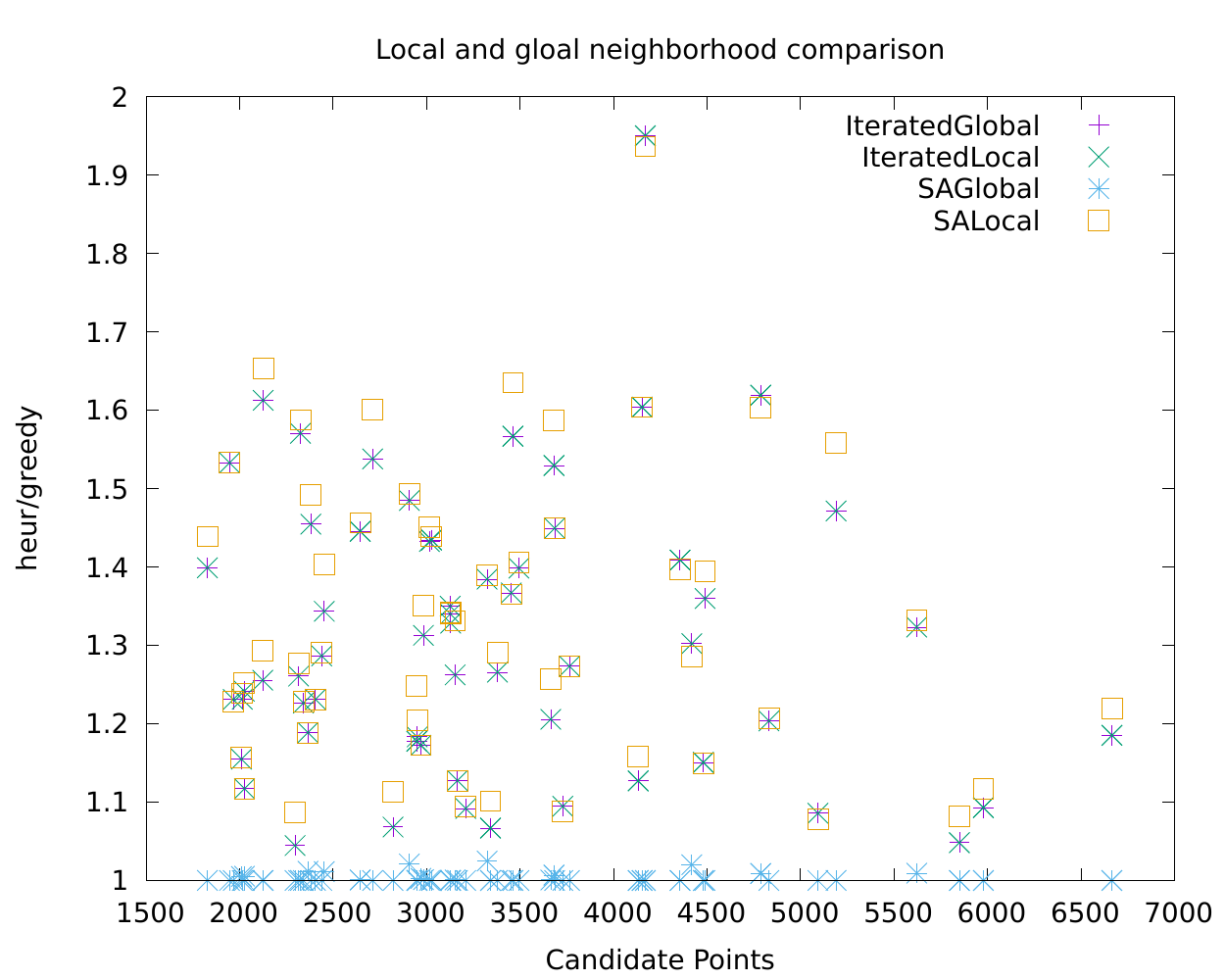}
\subcaption{Comparison of achieved objective values and greedy solutions}
\label{fig:NeighCompSol}
\end{subfigure}
\caption[Comparison local and global neighborhood ]{Comparison of solutions for local and global neighborhood with Iterated Local Search and Simulated Annealing for $k = 25$.}
\label{fig:NeighComp}
\end{figure}

Results are shown in Figure~\ref{fig:NeighComp}. 
When using a global neighborhood,
it turns out that in most cases, Simulated Annealing does not find better
solutions than those provided by greedy. This is most likely due to the fact
that the neighborhood space is too large: Looking for a random neighbor will
result 
in a swapped portal that is not connected to another portal
via segments, yielding a worse solution.
Therefore, the objective value is lowered at high temperatures,
while better neighboring solutions will be hard to find at lower temperatures.
This is illustrated in Figure~\ref{fig:NeighCompTime}, where runtimes 
are lower than for any other heuristic, with the exception of greedy.
This is because the algorithm terminates if no improvements of the best 
solution are found for a certain number of iterations.  As a result, global
neighborhoods are not recommended for Simulated Annealing.

On the other hand, the results with local neighborhoods are better with respect
to the objective value. However, Simulated Annealing with local
neighborhood also requires the largest runtimes by a wide margin, 
in particular for large instances,
as shown in Figure~\ref{fig:NeighCompTime}.
This can be explained by the ``reheating'' process after low temperatures:
The algorithm tries to find better solutions than the current solution for a while;
if no better solution is found for a fixed number of iterations,
 it reheats the temperature to escape a local optimum,
allowing worse solutions as the current one. As the process continues
when an improvement is found, the resulting gradual search process
may continue for a long time.

Both the local and the global variants of Iterated Local Search produce
almost always the same solution values, as shown in Figure~\ref{fig:NeighCompSol}.
This is due to the fact that both are looking for the best neighbor in every 
iteration. Because a swapped portal needs to be connected to one or more 
portals via segments, the local neighborhood consists of all vertices that are connected to at least one portal via segments.
Therefore, both search variants produce the same solutions almost every time.
This makes local Iterated Local Search slightly preferable:
As shown in Figure~\ref{fig:NeighCompTime}, it is marginally faster
than global Iterated Local Search by a small margin, due to smaller search space.
(This effect is not more pronounced because the search for a locally best neighbor
seems to be almost as slow as for a globally best neighbor.)

Overall local neighborhoods appear to perform better than global ones,
which is why they were selected for the following tests.

\subsection{An IP Example}
\label{sec:example}
Figure~\ref{fig:segments_to_graph} shows an example for the IP formulated in Section~\ref{sec:IP}; 
using ``$x_i$'' as shorthand for the IP variables $x_{\tau,v_i v_{i+1}}$ that
correspond to the edges along trajectory path
$\tau=(v_0,v_1,\ldots,v_6)$.  Constraints~2 are:
$x_0\leq y_0$, $x_1\leq y_1+x_0$, $x_2\leq y_2+x_1, \ldots, x_5\leq
y_5+x_4$.  Constraints~3 are: $x_5\leq y_6$,
$x_4\leq y_5+x_5$, $x_3\leq y_4+x_4, \ldots,x_0\leq y_1+x_1$.  With
portals selected at $v_1$ and $v_4$ (i.e., with $y_1=y_4=1$, $y_i=0$
otherwise), we get constraints $x_0\leq 0$, $x_5\leq 0$, and $x_4\leq
y_5+x_5=0$, implying that only the subpath $(v_1,v_2,v_3,v_4)$ of
$\tau$ contributes to the objective function.

\begin{figure}[h]
  \centering
  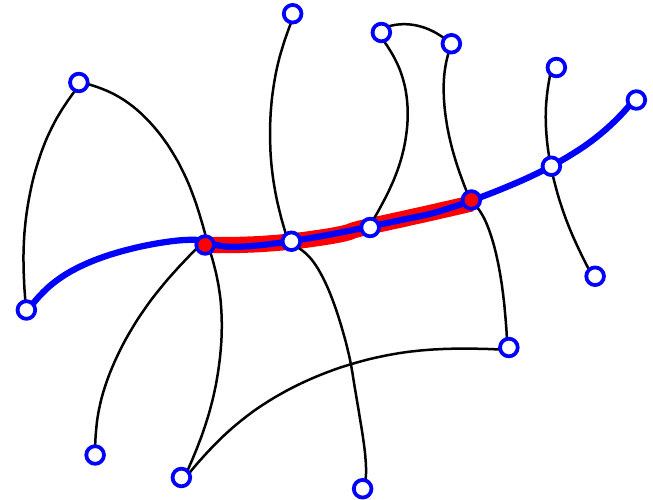
  \caption{Formulating the IP: Trajectory $\tau=(v_0,v_1,\ldots,v_6)$ is highlighted in blue. With selected portals at $v_1$ and $v_4$, the portion highlighted in red is captured.}
  \label{fig:segments_to_graph}
\end{figure}

\subsection{Additional Images and Plots}

Additional images illustrating instances and results from the experimental Section~\ref{sec:practice} are shown
in Figures~\ref{fig:huge}--\ref{fig:TANOC25_large}.

\begin{figure}
\centering
\includegraphics[width = 0.49\textwidth]{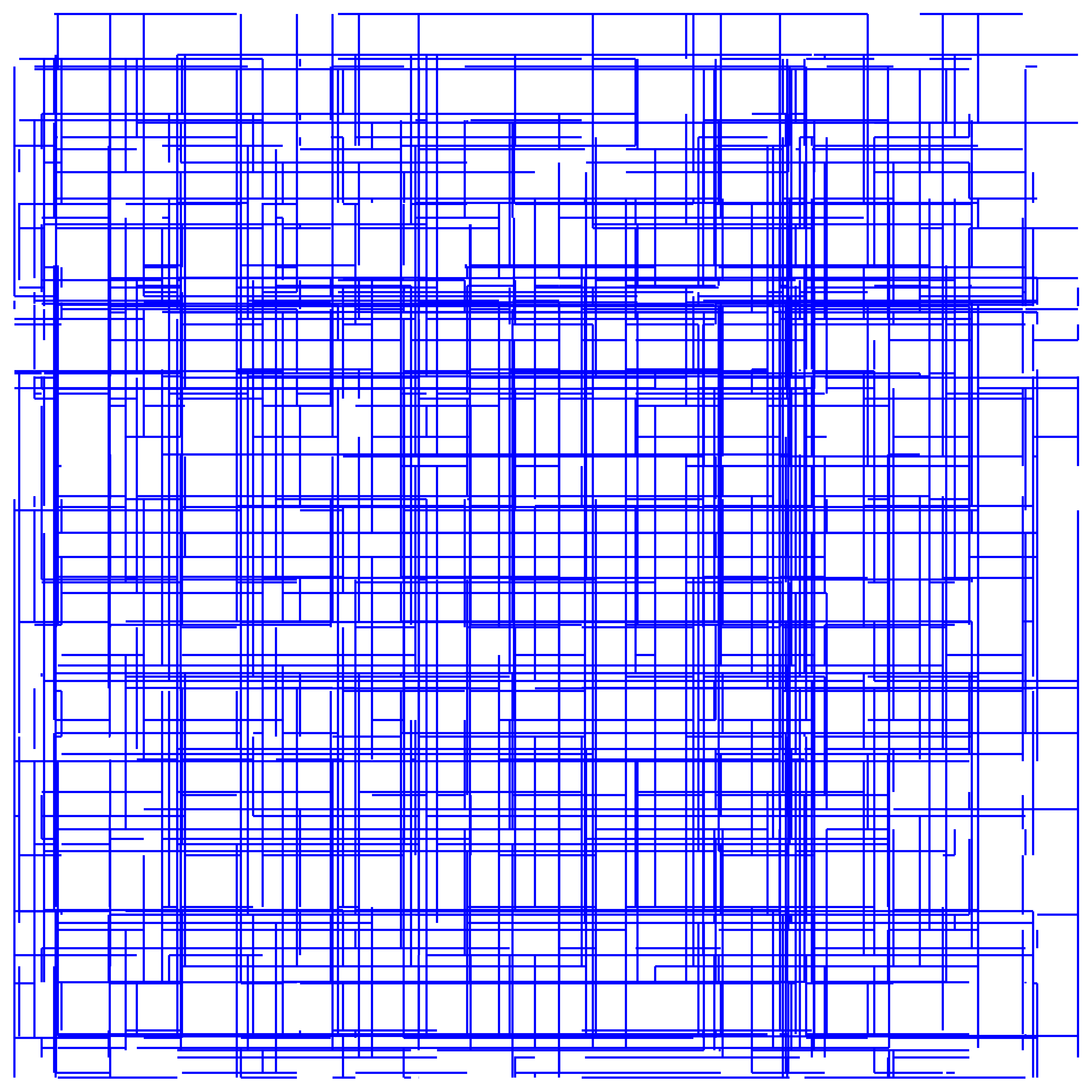}
\caption{An instance with 1100 non-degenerate axis-parallel segments}
\label{fig:huge}
\end{figure}

\begin{figure}[!h]
\centering
	\begin{subfigure}[t]{0.49\textwidth}
	\centering
	\includegraphics[width=\textwidth]{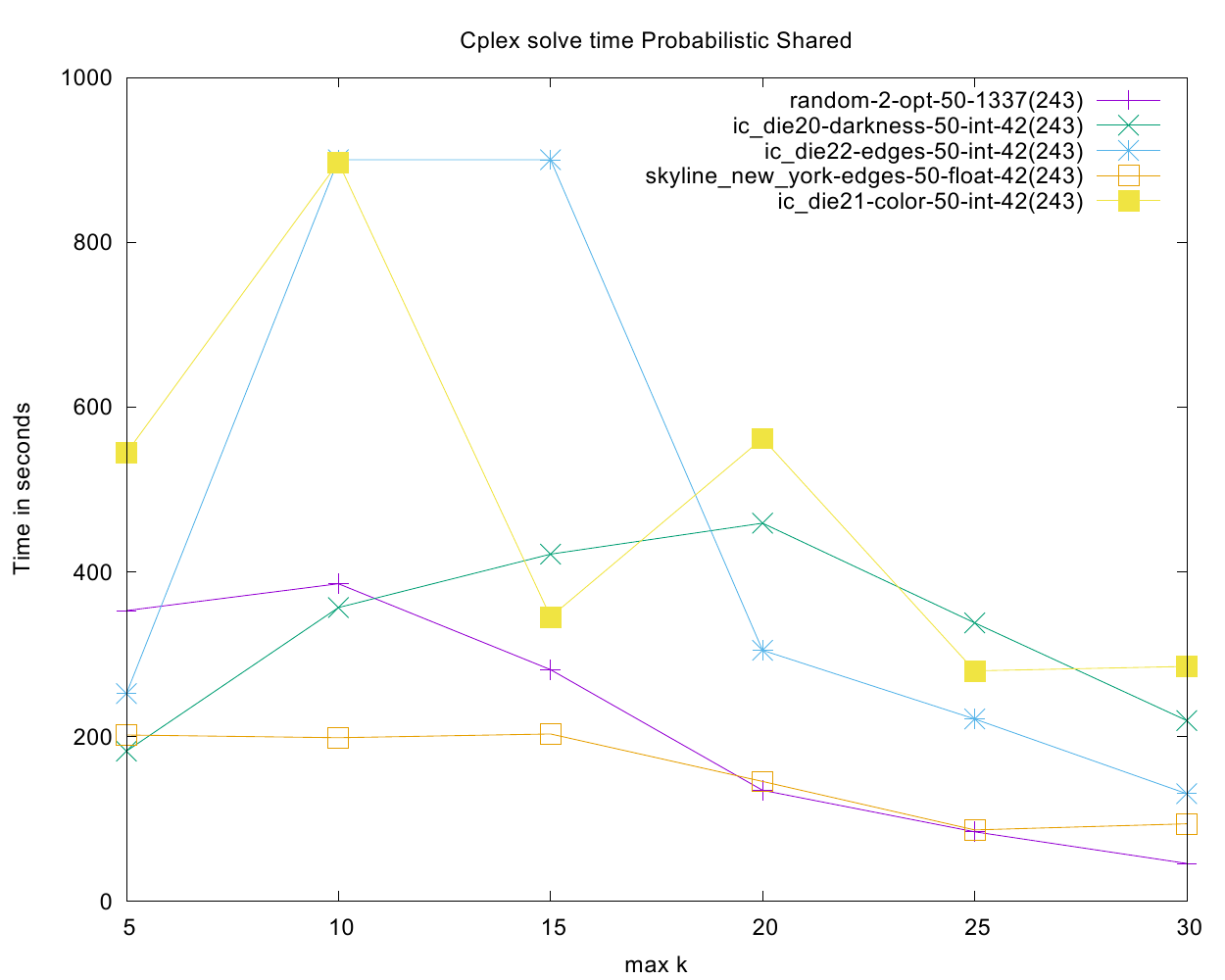}
	\subcaption{Graphs with 50 seed points}
\label{fig:ProbSharedS50_L}
\end{subfigure}
\begin{subfigure}[t]{0.49\textwidth}
	\centering
	\includegraphics[width=\textwidth]{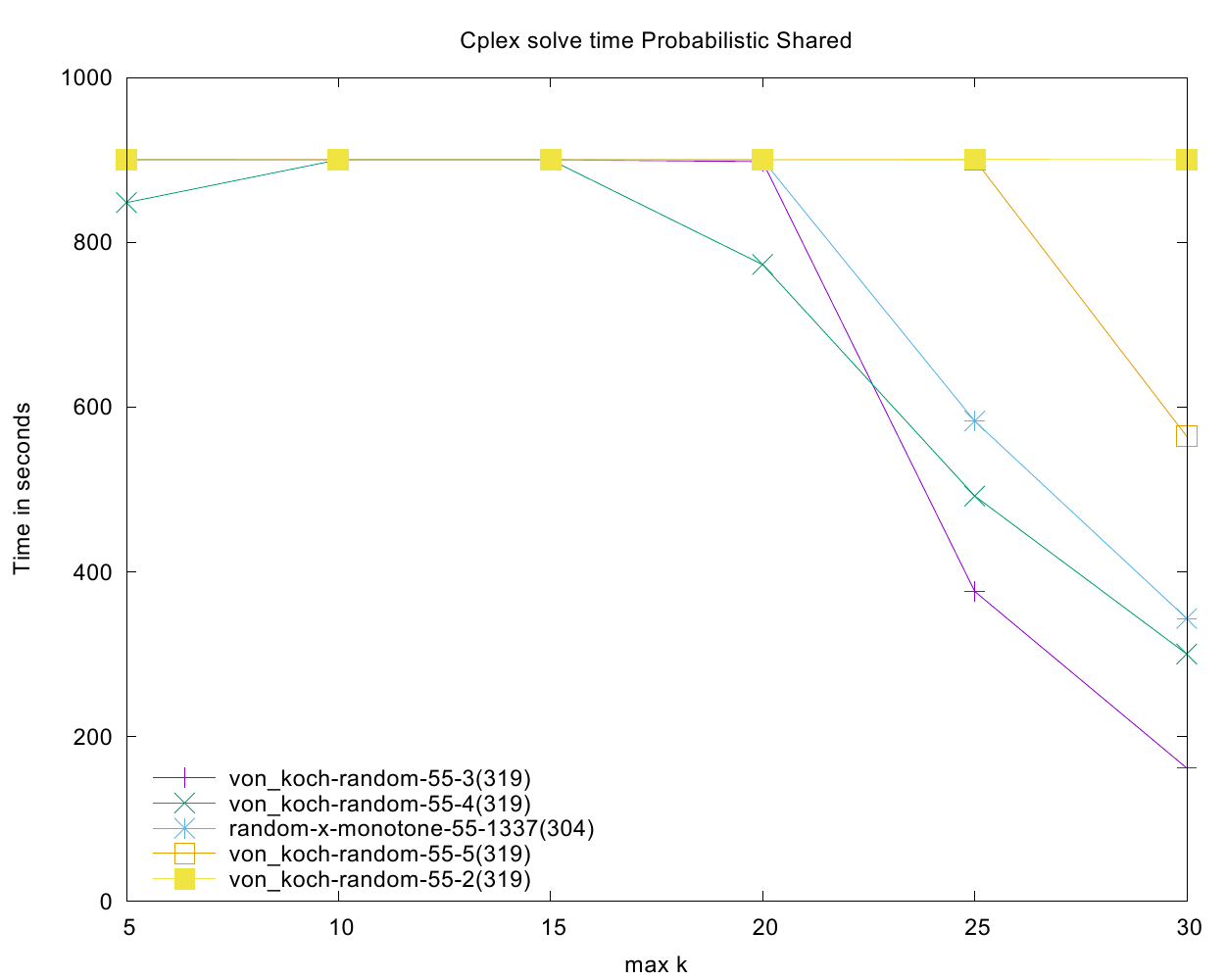}
	\subcaption{Graphs with 55 seed points}
\label{fig:ProbSharedS55_L}
\end{subfigure}
\caption[Probabilistic generator instances solving times (large)]{Solution times for generated
  instances, with a time limit of 900 seconds for every instance}
\label{fig:ProbShared_L}
\end{figure}

\begin{figure}
\centering
\begin{subfigure}[t]{0.49\textwidth}
\includegraphics[width=\textwidth]{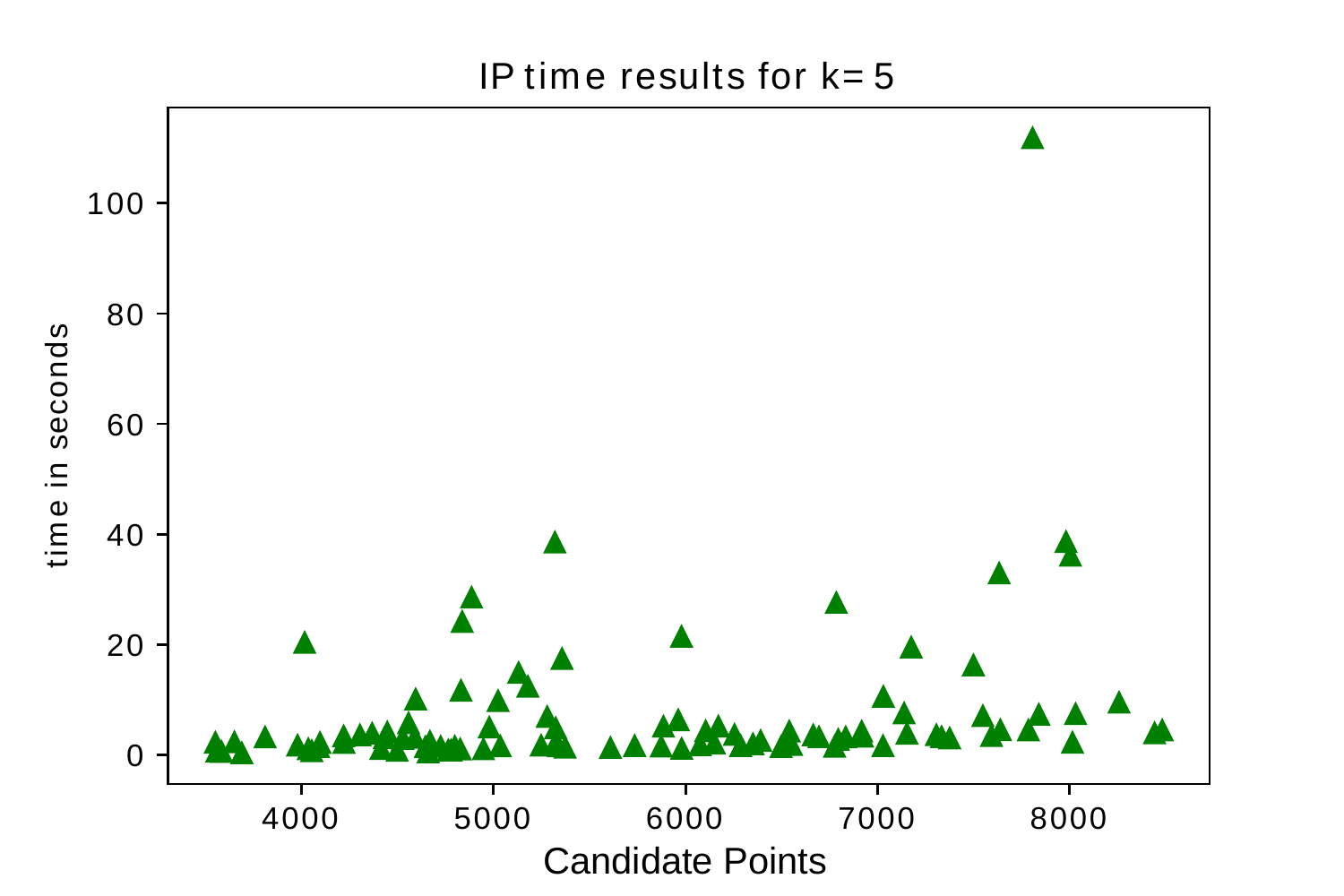}
\subcaption{IP runtime for $k=5$}\label{fig:TANOC_time5}
\end{subfigure}
\begin{subfigure}[t]{0.49\textwidth}
\includegraphics[width=\textwidth]{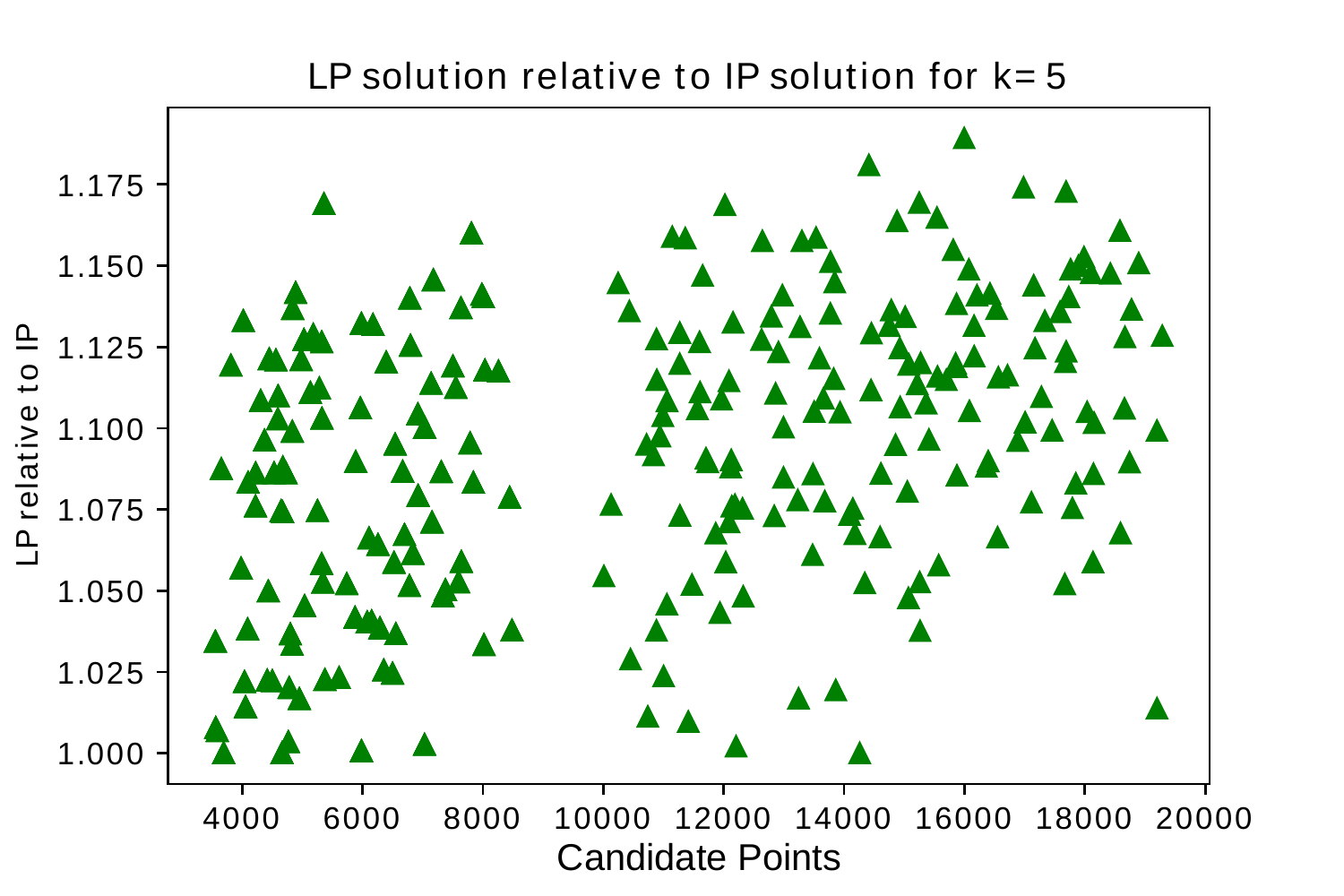}
\subcaption{Integrality gap for $k=5$}\label{fig:TANOC_gap5}
\end{subfigure}
\caption[Comparison solved algorithms]{IP runtime and integrality gap for axis-parallel instances and $k=5$.}
\label{fig:TANOC5_large}
\end{figure}

\begin{figure}
\centering
\begin{subfigure}[t]{0.49\textwidth}
\includegraphics[width=\textwidth]{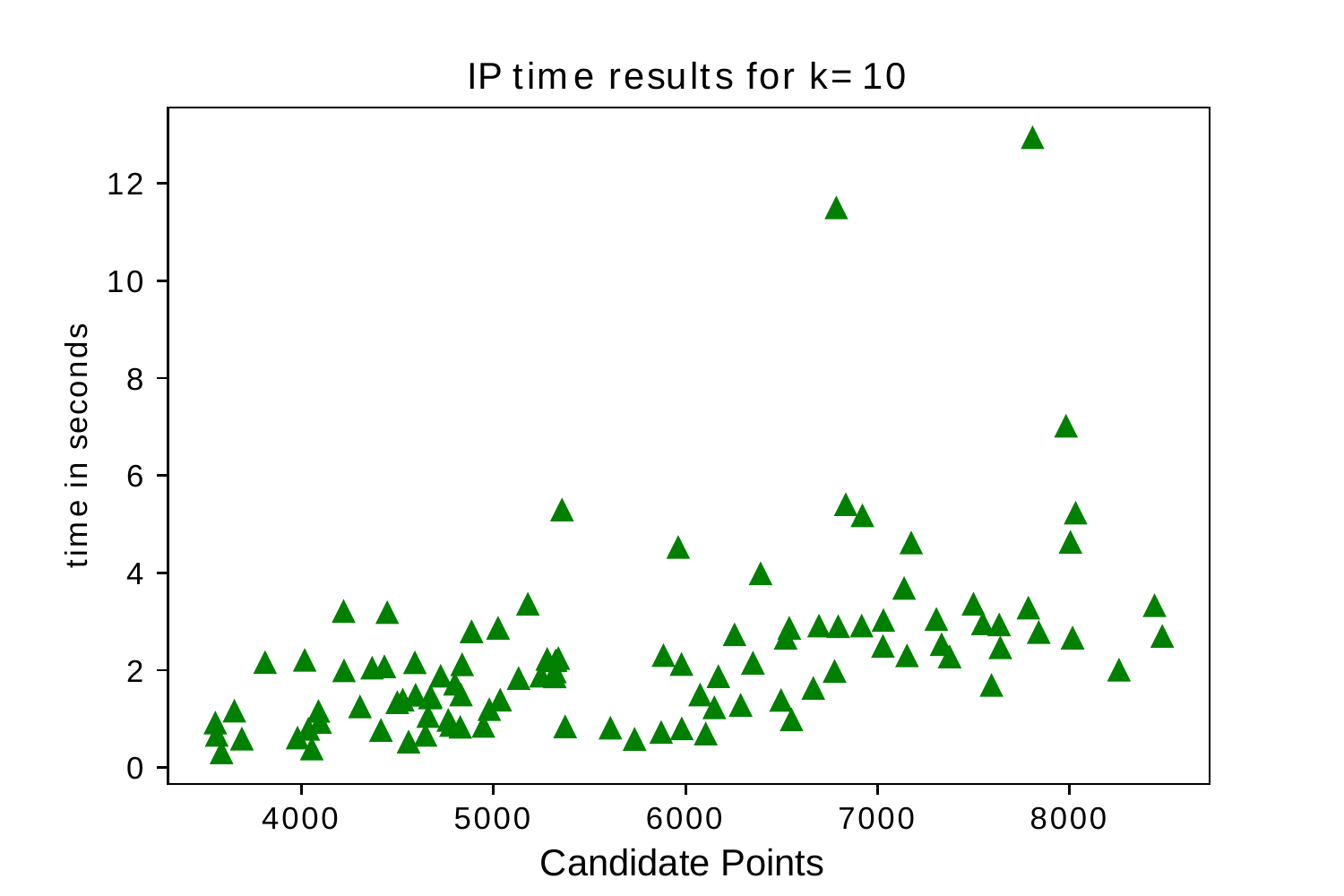}
\subcaption{IP runtime for $k=10$}\label{fig:TANOC_time10}
\end{subfigure}
\begin{subfigure}[t]{0.49\textwidth}
\includegraphics[width=\textwidth]{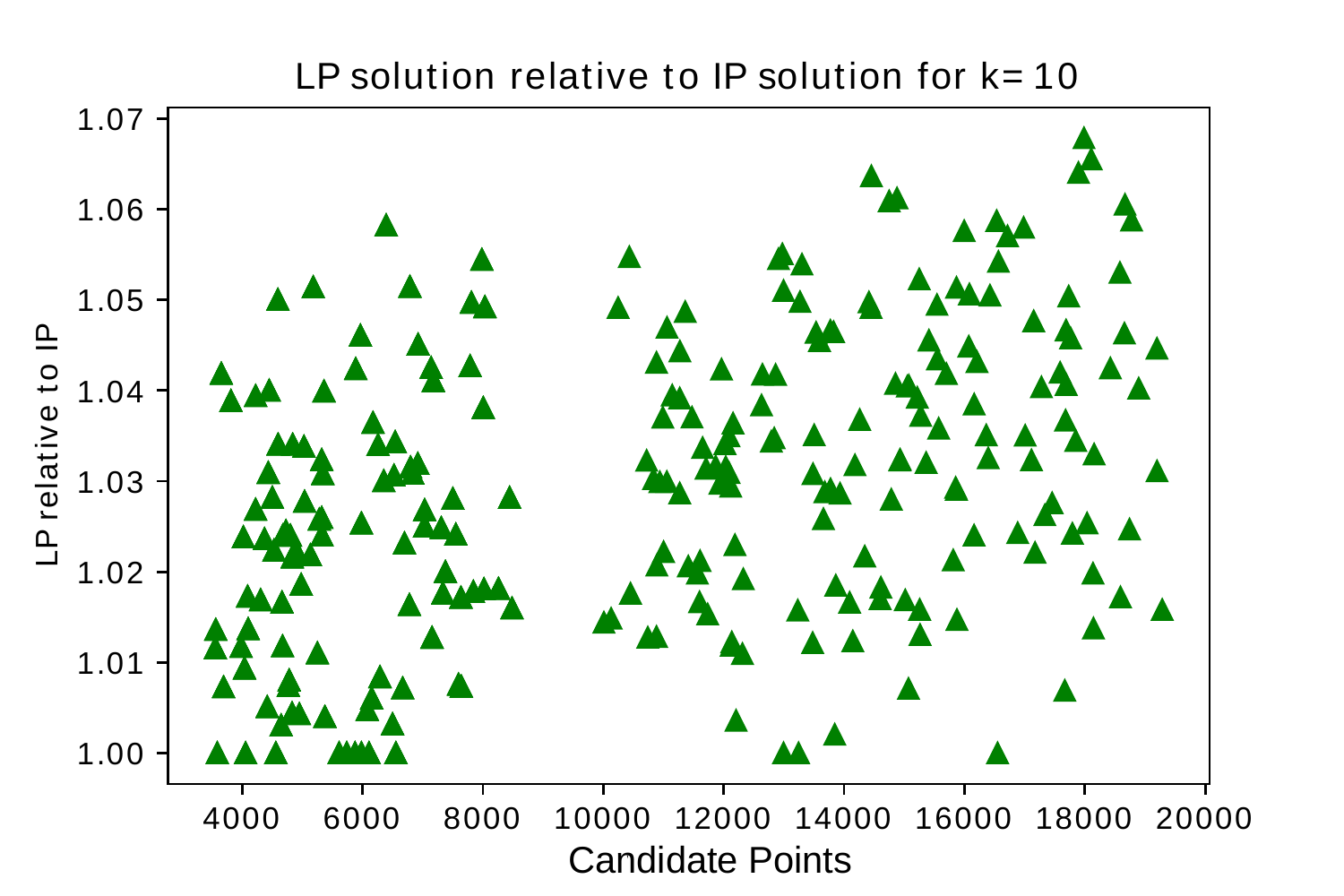}
\subcaption{Integrality gap for $k=10$}\label{fig:TANOC_gap10}
\end{subfigure}
\caption[Comparison solved algorithms]{IP runtime and integrality gap for axis-parallel instances and $k=10$.}
\label{fig:TANOC10_large}
\end{figure}

\begin{figure}
\centering
\begin{subfigure}[t]{0.49\textwidth}
\includegraphics[width=\textwidth]{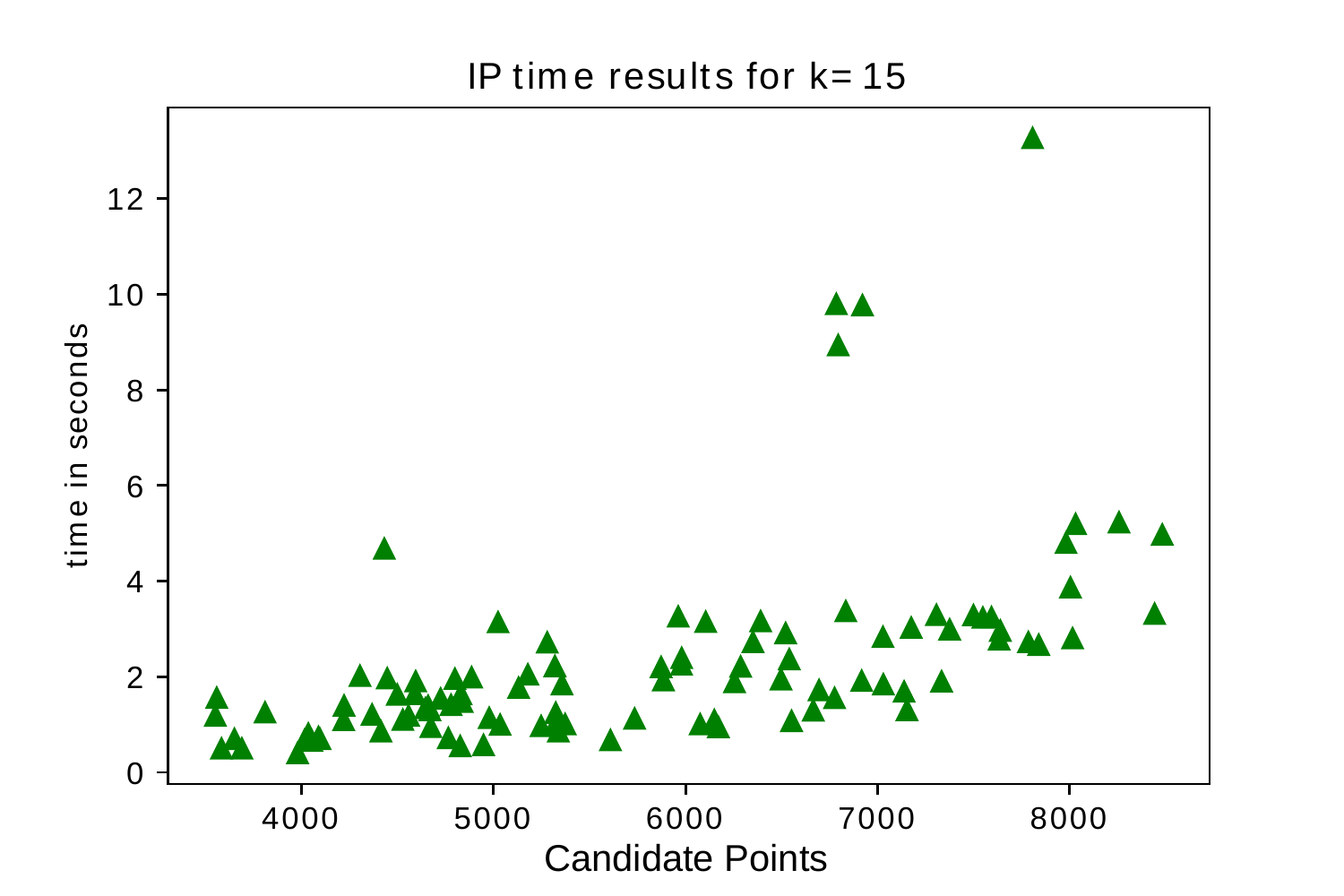}
\subcaption{IP runtime for $k=15$}\label{fig:TANOC_time15}
\end{subfigure}
\begin{subfigure}[t]{0.49\textwidth}
\includegraphics[width=\textwidth]{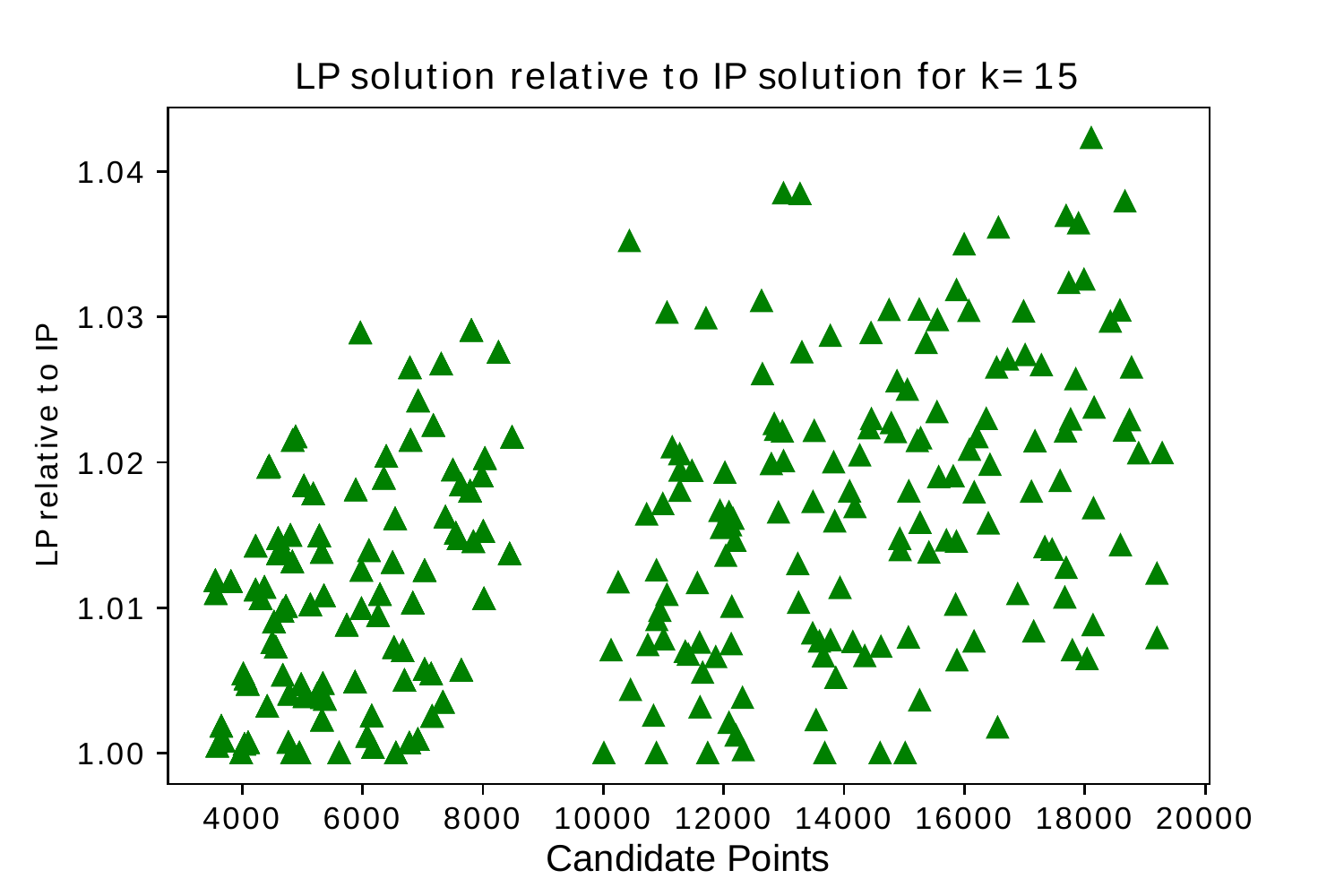}
\subcaption{Integrality gap for $k=15$}\label{fig:TANOC_gap15}
\end{subfigure}
\caption[Comparison solved algorithms]{IP runtime and integrality gap for axis-parallel instances and $k=15$.}
\label{fig:TANOC15_large}
\end{figure}

\begin{figure}
\centering
\begin{subfigure}[t]{0.49\textwidth}
\includegraphics[width=\textwidth]{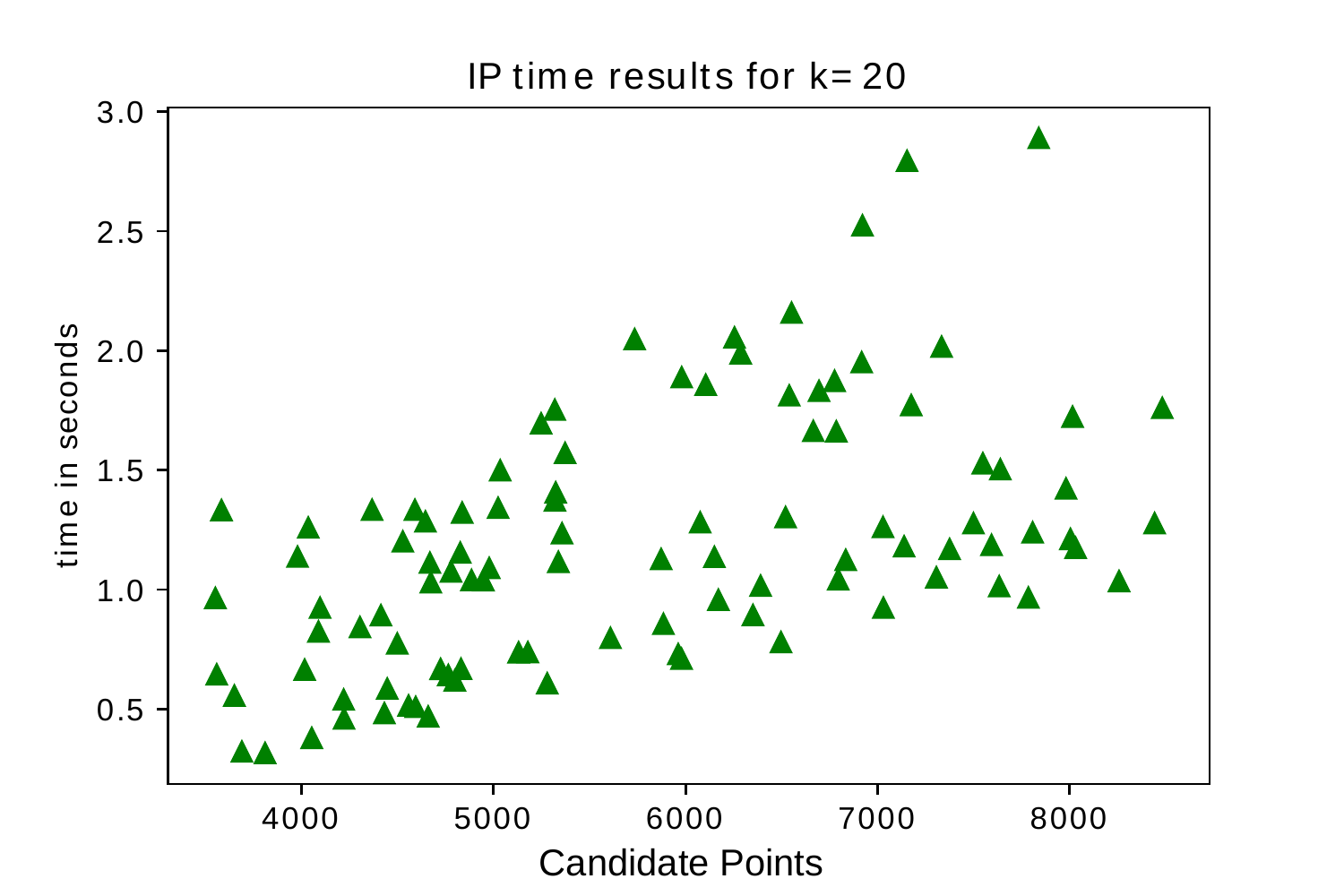}
\subcaption{IP runtime for $k=20$}\label{fig:TANOC_time20}
\end{subfigure}
\begin{subfigure}[t]{0.49\textwidth}
\includegraphics[width=\textwidth]{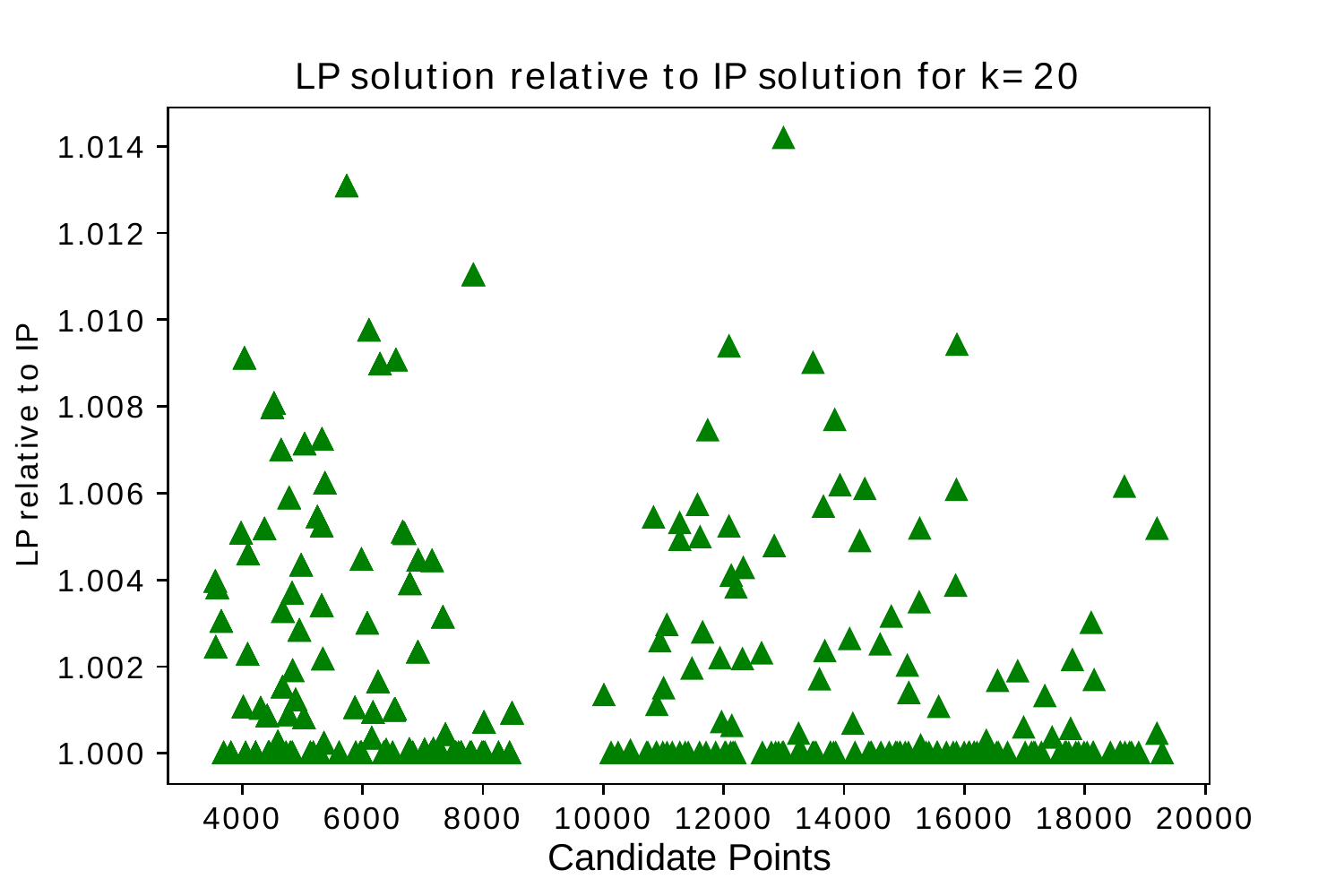}
\subcaption{Integrality gap for $k=20$}\label{fig:TANOC_gap20}
\end{subfigure}
\caption[Comparison solved algorithms]{IP runtime and integrality gap for axis-parallel instances and $k=20$.}
\label{fig:TANOC20_large}
\end{figure}

\begin{figure}
\centering
\begin{subfigure}[t]{0.49\textwidth}
\includegraphics[width=\textwidth]{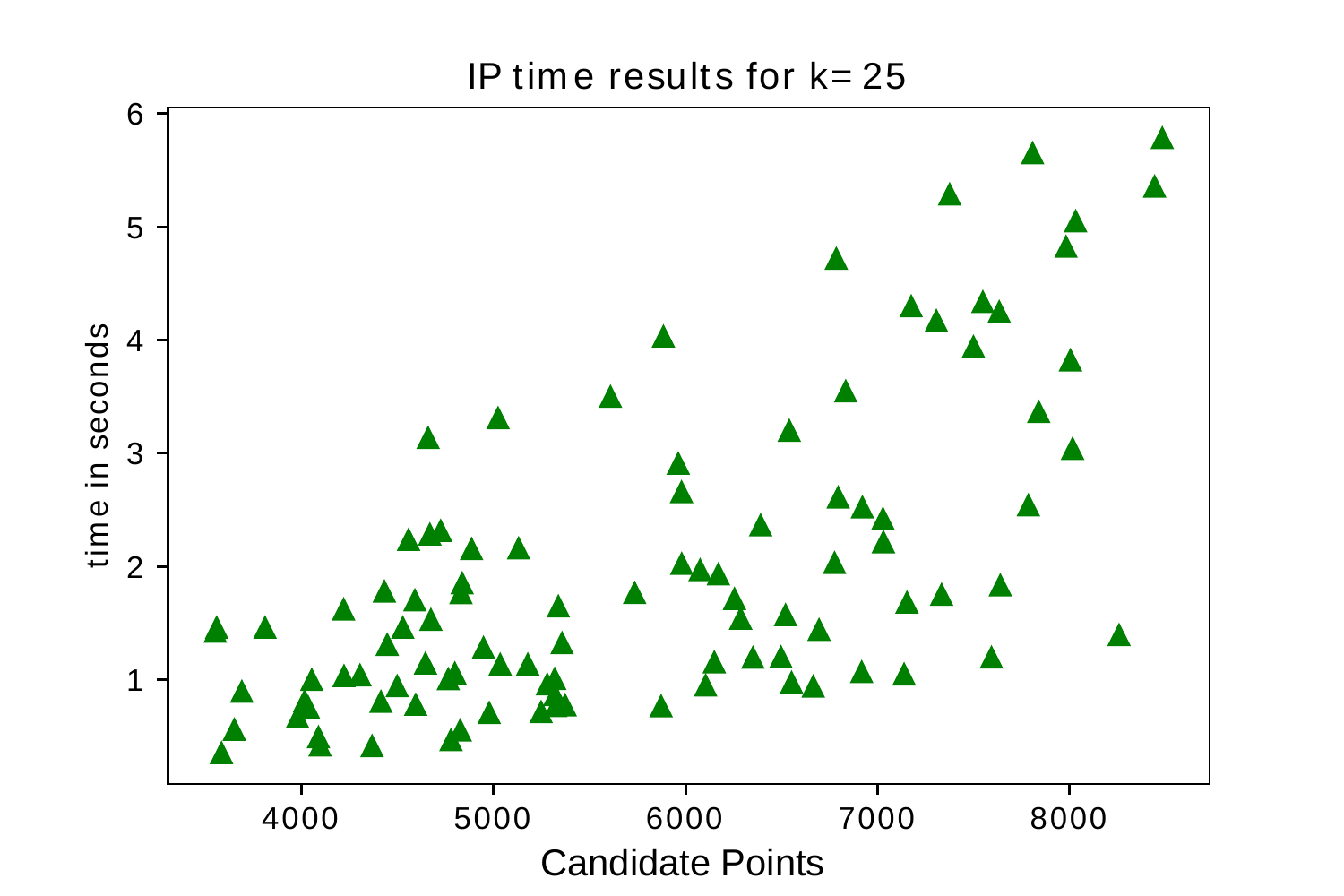}
\subcaption{IP runtime for $k=25$}\label{fig:TANOC_time25}
\end{subfigure}
\begin{subfigure}[t]{0.49\textwidth}
\includegraphics[width=\textwidth]{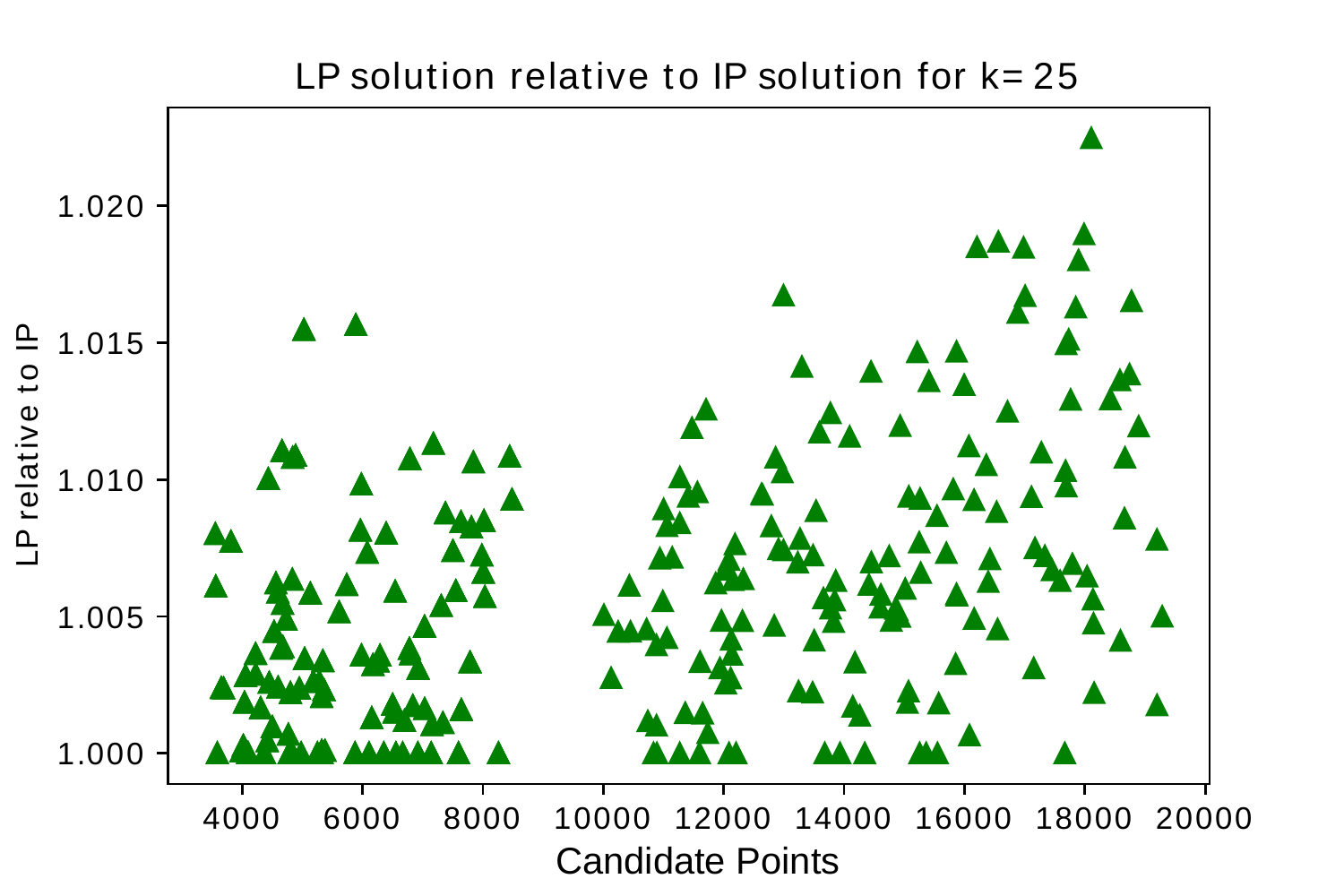}
\subcaption{Integrality gap for $k=25$}\label{fig:TANOC_gap25}
\end{subfigure}
\caption[Comparison solved algorithms]{IP runtime and integrality gap for axis-parallel instances and $k=25$. \\
\vspace*{6cm}
\ 
}
\label{fig:TANOC25_large}
\end{figure}

\end{document}